\documentclass[a4paper,UKenglish,cleveref,autoref,thm-restate,pdfa,svgnames]{lipics-v2021}

\pdfoutput=1

\title{Expectation in Stochastic Games with Prefix-independent Objectives}

\author{Laurent Doyen}
{Universit\'{e} Paris-Saclay, CNRS, ENS Paris-Saclay, LMF, Gif-sur-Yvette, France}
{ldoyen@lmf.cnrs.fr}
{https://orcid.org/0000-0003-3714-6145}
{}

\author{Pranshu Gaba}
{Tata Institute of Fundamental Research, Mumbai, India}
{pranshu.gaba@tifr.res.in}
{https://orcid.org/0009-0000-8012-780X}
{}

\author{Shibashis Guha}
{Tata Institute of Fundamental Research, Mumbai, India}
{shibashis@tifr.res.in}
{https://orcid.org/0000-0002-9814-6651}
{}

\authorrunning{L. Doyen, P. Gaba, and S. Guha}

\Copyright{Laurent Doyen, Pranshu Gaba, and Shibashis Guha}



\relatedversion{Full version of paper accepted in CONCUR 2025.}

\nolinenumbers 

\hideLIPIcs

\usepackage{mypackages}
\usepackage{mymacros}

\IfKnowledgePaperModeTF{}{
  \definecolor{abc}{HTML}{300030}
  \definecolor{def}{HTML}{1c2f2f}
  \knowledgestyle{intro notion}{color={def}, emphasize}
  \knowledgestyle{notion}{color={abc}}
    \hypersetup{
    colorlinks=true,
    breaklinks=true,
    linkcolor={},
    citecolor={},
    urlcolor={lipicsGray},
    }
}

\knowledge {notion}
    |  probability distribution

\knowledge {notion}
    |  stochastic game
    |  stochastic games
    |  Stochastic game
    |  Stochastic games

\knowledge {notion}
    |  \VerticesMain

\knowledge {notion}
    |  \VerticesAdversary

\knowledge {notion}
    |  probabilistic vertices
    |  probabilistic vertex
    |  \VerticesRandom
    
\knowledge {notion}
    |  probability function

\knowledge {notion}
    |  payoff function

\knowledge {notion}
    |  prefix
    |  prefixes

\knowledge {notion}
    |  play
    |  plays

\knowledge {notion}
    |  trap

\knowledge {notion}
    |  suffix
    |  suffixes

\knowledge {notion}
    |  Markov decision process
    |  Markov decision processes
    |  Markov decision process (MDP)
    |  MDP
    |  MDPs

\knowledge {notion}
    |  Markov chain

\knowledge {notion}
    |  non-stochastic game
    |  non-stochastic games
    |  non-stochastic two-player game

\knowledge {notion}
    |  simple stochastic game
    |  simple stochastic games
    |  Simple stochastic game
    |  Simple stochastic games
    |  SSG
    |  SSGs

\knowledge {notion}
    |  Boolean objective
    |  Boolean objectives

\knowledge {notion}
    |  quantitative objective
    |  quantitative objectives
    |  Quantitative objective
    |  Quantitative objectives
    |  objective@quantitative

\knowledge {notion}
    |  threshold objective
    |  threshold objectives
    |  threshold Boolean objective
    |  threshold Boolean objectives

\knowledge {notion}
    |  prefix-independent
    |  prefix independence
    |  Prefix independence
    |  Prefix-independent

\knowledge {notion}
    |  Bellman equations

\knowledge {notion}
    |  reachability
    |  reachability objective
    |  \ReachObj

\knowledge {notion}
    |  strategy

\knowledge {notion}
    |  memory

\knowledge {notion}
    |  memoryless

\knowledge {notion}
    |  boundary vertex
    |  boundary vertices

\knowledge {notion}
    |  value vector
    |  value vectors
    |  vector@value

\knowledge {notion}
    | class
    | classes

\knowledge {notion}
    |  trap subgame
    |  trap subgames
    
\knowledge {notion}
    |  almost-sure satisfaction problem

\knowledge {notion}
    |  expectation problem

\begin{document}

\maketitle

\begin{abstract}
    Stochastic two-player games model systems with an environment that is both adversarial and stochastic. 
    In this paper, we study the expected value of bounded quantitative \kl{prefix-independent} objectives in the context of \kl{stochastic games}.
    We show a generic reduction from the \kl{expectation problem} to linearly many instances of the almost-sure satisfaction problem for \kl{threshold Boolean objectives}.
    The result follows from partitioning the vertices of the game into so-called value classes where each class consists of vertices of the same value.
    Our procedure further entails that the memory required by both players to play optimally for the \kl{expectation problem} is no more than the memory required by the players to play optimally for the \kl{almost-sure satisfaction problem} for a corresponding \kl{threshold Boolean objective}.

    We show the applicability of the framework to compute the expected window mean-payoff measure in \kl{stochastic games}.
    The window mean-payoff measure strengthens the classical mean-payoff measure by computing the mean payoff over windows of bounded length that slide along an infinite path.
    We show that the decision problem to check if the expected window mean-payoff value is at least a given threshold is in \(\UP \intersection \coUP\) when the window length is given in unary.
\end{abstract}

\section{Introduction}
Reactive systems typically have an infinite execution where the controller continually reacts to the environment.
Given a specification, the \emph{reactive controller synthesis} problem~\cite{Chu63} concerns with synthesising a policy for the controller such that the specification is satisfied by the system for all behaviours of the environment. 
This problem is modelled using two-player turn-based games on graphs, where the two players are the controller (\(\PlayerMain\)) and the environment (\(\PlayerAdversary\)), the vertices and the edges of the game graph represent the states and transitions of the system, and the objective of \(\PlayerMain\) is to satisfy the specification. 
An execution of the system is then an infinite path in the game graph. 
The reactive controller synthesis problem corresponds to determining if there exists a strategy of \(\PlayerMain\) such that for all strategies of \(\PlayerAdversary\), the outcome satisfies the objective. 
If such a winning strategy exists, then we would also like to synthesize it.
The environment is considered as an adversarial player to ensure that the specification is met even in the worst-case scenario. 

Objectives are either Boolean or quantitative. 
Each execution either satisfies a \kl{Boolean objective} \(\BooleanObjective\) or does not satisfy \(\BooleanObjective\). 
The set of executions that satisfy \(\BooleanObjective\) form a language over infinite words with the alphabet being the set of vertices in the graphs.
On the other hand, a \kl{quantitative objective} \(\Objective\) evaluates the performance of the execution by a numerical metric, which \(\PlayerMain\) aims to maximize and \(\PlayerAdversary\) aims to minimize.
A \kl{quantitative objective} can be viewed as a real-valued function over infinite paths in the graph.

In the presence of uncertainty or probabilistic behaviour, the game graph becomes stochastic.
Fixing the strategies of the two players gives a distribution over infinite paths in the game graph.
For \kl{Boolean objectives} \(\BooleanObjective\), the goal of \(\PlayerMain\) is to maximize the probability that an outcome satisfies \(\BooleanObjective\). 
For \kl{quantitative objectives} \(\Objective\), there are two possible views: 
\begin{description}
    \item[Satisfaction.] Given a threshold \(\Threshold\), to maximize the probability that \(\Objective\)-value of the outcome is greater than \(\Threshold\);
    \item[Expectation.] To maximize the \(\Objective\)-value of the outcome in expectation. 
\end{description}
Either view may be desirable depending on the context~\cite{BBCFK14,BFRR17,BDOR20,BGR19}.
The satisfaction view can be seen as a \kl{Boolean objective}: the \(\Objective\)-value of the outcome is either greater than \(\Threshold\) or it is not.
The expectation view is more nuanced, and is the subject of study in this paper. 

In this paper, we look at the \kl{expectation problem} for quantitative \kl{prefix-independent} objectives (also referred to as tail objectives).
These are objectives that do not depend on finite prefixes of the \kl{plays}, but only on the long-run behaviour of the system.
In systems, we are often willing to allow undesirable behaviour in the short-term, if the long run behaviour is desirable. 
\kl{Prefix-independent} objectives model such requirements and thus are of interest to study~\cite{Cha07}.
\kl{Prefix-independent} objectives also have the benefit that they satisfy the \kl{Bellman equations}~\cite{KMW23}, which simplifies their analysis. 
The \kl{expectation problem} for such objectives arises naturally in many scenarios. 
For example:
\begin{romanenumerate}
    \item An algorithmic trading system is designed to generate profit by executing trades based on real-time market data. 
    Following an initial phase of learning and unstable behaviour due to parameter tuning, average profit over a bounded time window must always exceed a threshold and decisions need to be made within short well-defined intervals for them to be effective.
    \item A power plant may have different strategies to produce power (such as coal, solar, nuclear, wind) and must allocate resources among these strategies so as to maximize the power produced in expectation.
\end{romanenumerate}

\subparagraph*{Contributions.}
All of our contributions are with regard to quantitative \kl{prefix-independent} objectives \(\Objective\) that are bounded (i.e., the image of \(\Objective\) is bounded between integers \(-\ObjectiveBound{\Objective}\) and \(\ObjectiveBound{\Objective}\)) and such that a bound \(\DenBound{\Objective}\) on the denominators of the optimal expected \(\Objective\)-values of vertices in the game is known.
The bound on the image ensures determinacy~\cite{Mar98}, that is, the players have optimal strategies, and the bound on the denominator of optimal values of vertices discretize the search space. 
These bounds often exist and are easily derivable for common objectives of interest such as mean payoff.

Our primary contribution is a reduction of the \kl{expectation problem} for such an objective \(\Objective\) to linearly many instances of the \kl{almost-sure satisfaction problem} for \kl{threshold Boolean objectives} \(\ThresholdObjective{> \Threshold}\) for thresholds \(\Threshold \in \Rationals\).
Deciding the almost-sure satisfaction of \(\ThresholdObjective{> \Threshold}\) is conceptually simpler than computing the expected value of \(\Objective\), as in the former, we only need to consider if the measure of the paths that satisfy the objective \(\ThresholdObjective{> \Threshold}\) is equal to one, whereas in the latter, one must take the averages of the measures of the sets of paths \(\Play\) weighted with the value \(\Objective(\Play)\) of the paths.
Our technique is generic in the sense that when an algorithm for the \kl{almost-sure satisfaction problem} for \(\ThresholdObjective{> \Threshold}\) is known, we directly obtain the complexity and a way to solve the \kl{expectation problem} for \(\Objective\).

Our reduction builds on the technique introduced in~\cite{CHH09} for Boolean \kl{prefix-independent} objectives and non-trivially extends it to  quantitative \kl{prefix-independent} objectives \(\Objective\) for which the bounds \(\ObjectiveBound{\Objective}\) and \(\DenBound{\Objective}\) are known.
The expected \(\Objective\)-values of vertices are nondeterministically guessed, and we present a characterization~(\Cref{thm:six-conditions}, similar to~\cite[Lemma~8]{CHH09} but with important and subtle differences) to verify the guess.
We also explicitly construct strategies for both players that are optimal for the expectation of \(\Objective\), in terms of almost-sure winning strategies for \(\ThresholdObjective{> \Threshold}\) (proof of \Cref{lem:fwmpl-almost-sure-to-expected}). 
The memory requirement for the constructed optimal strategies is the same as that of the almost-sure winning strategies~(\Cref{cor:memory-bound}).

Our framework gives an alternative approach to solve the \kl{expectation problem} for well-studied objectives such as expected mean payoff and gives new results for not-as-well-studied objectives such as the \emph{window mean-payoff objectives} introduced in~\cite{CDRR15}.
As our secondary contribution, we illustrate our technique by applying it to two variants of window mean-payoff objectives: fixed (\(\FWMPL\)) and bounded (\(\BWMP\)) window mean-payoff.
Using our reduction, we are able to show that for both of these objectives, the expectation problem is in \(\UP \intersection \coUP\) (\Cref{thm:fwmp-summary} and \Cref{thm:bwmp-summary}), a result that was not known before.
The \(\UP \intersection \coUP\) upper bound for window mean-payoff objectives matches the special case of \kl{simple stochastic games}~\cite{Con92, CF11}, and thus would require a major breakthrough to be improved.
The lower bounds on the memory requirements for these objectives carry over from special case of the \kl{non-stochastic games}~\cite{CDRR15,DGG25LMCS}.
We summarize the complexity results and bounds on the memory requirements for the window mean-payoff objectives in \Cref{tab:results-summary}. 

\begin{table}[t]
    \caption{Complexity and bounds on memory requirement for window mean-payoff objectives}%
    \label{tab:results-summary}
    \begin{tabular}{lccc}
        \toprule
        \thead{Objective}  & \thead{Complexity} & \thead{Memory (\(\PlayerMain\))\\ (lower~\cite{DGG25LMCS}, upper)} & \thead{Memory (\(\PlayerAdversary\))\\ (lower~\cite{DGG25LMCS}, upper)}  \\
        \midrule
        \(\FWMPL\) & \(\UP \intersection \coUP\) & \(\WindowLength - 1\), \(\WindowLength\) & \(\abs{\Vertices} - \WindowLength\), \(\abs{\Vertices} \cdot \WindowLength \) \\
        \(\BWMP\)  & \(\UP \intersection \coUP\) & memoryless, memoryless & infinite, infinite \\
        \bottomrule
    \end{tabular}
\end{table}

\subparagraph*{Related work.}
\kl{Stochastic games} were introduced by Shapley~\cite{Shapley53} where these games were studied under expectation semantics for discounted-sum objectives.
In~\cite{CH08b}, it was shown that solving stochastic parity games reduces to solving stochastic mean-payoff games.
Further, solving stochastic parity games, stochastic mean-payoff games, and \AP \intro{simple stochastic games} (i.e., \kl{stochastic games} with \kl{reachability} objective) are all polynomial-time equivalent~\cite{GM08,AM09}, and thus, are all in \(\UP \intersection \coUP\)~\cite{CF11}.
A sub-exponential (or even quasi-polynomial) time deterministic algorithms for \kl{simple stochastic games} on graphs with poly-logarithmic treewidth was proposed in~\cite{CMSS23}.
In~\cite{GK23}, sufficient conditions on the objective were shown such that optimal deterministic memoryless strategies exist for the players.
In~\cite{KMW23}, value iteration to solve the \kl{expectation problem} in \kl{stochastic games} with \kl{reachability}, safety, total-payoff, and mean-payoff objectives was studied.

Mean-payoff objectives were studied initially in two-player games, without stochasticity~\cite{EM79,ZP96}, and with stochasticity in~\cite{Gillette58}.
Finitary versions were introduced as window mean-payoff objectives~\cite{CDRR15}.
For finitary mean-payoff objectives, the satisfaction problem~\cite{BDOR20} and the \kl{expectation problem}~\cite{BGR19} were studied in the special case of \kl{Markov decision processes} (MDPs), which correspond to \kl{stochastic games} with a trivial adversary.
Expected mean payoff, expected discounted payoff, expected total payoff, etc., are widely studied for MDPs~\cite{Puterman94,BK08}.
Both the \kl{expectation problem}~\cite{BGR19} and the satisfaction problem~\cite{BDOR20} for the \(\FWMPL\) objective are in $\PTime$, while they are in \(\UP \cap \coUP\) for the \(\BWMP\) objective.
Ensuring the satisfaction and expectation semantics  simultaneously was studied in MDPs for the mean-payoff objective in~\cite{CKK17} and for the window mean-payoff objectives in~\cite{GG25}.
In both cases, the complexity was shown to be no greater than that of only expectation optimization.

The satisfaction problem for window mean-payoff objectives has been studied for two-player stochastic games in~\cite{DGG25LMCS}.
While positive and almost-sure satisfaction of \(\FWMPL\) are in $\PTime$, it follows from~\cite{DGG25LMCS} that the problem is in \(\UP \cap \coUP\) for quantitative satisfaction i.e., with threshold probabilities $0 < p < 1$.
Furthermore, the satisfaction problem of \(\BWMP\) is in \(\UP \cap \coUP\) and thus has the same complexity as that of the special case of MDPs~\cite{BDOR20}.

\subparagraph*{Outline.}
In \Cref{sec:preliminaries}, we give the necessary technical preliminaries.
In \Cref{sec:reducing-expectation-to-almost-sure-satisfaction}, we show a reduction from the expectation problem to the almost-sure satisfaction problem for prefix-independent objectives.
In \Cref{sec:window-mean-payoff-objectives}, we recall the window mean-payoff objectives and show an application of the reduction for these objectives and discuss their complexity. 
In particular, we give bounds on \(\DenBound{\Objective}\) for the \(\FWMPL\) objective in \Cref{sec:expected-fixed-window-mean-payoff-value} and for the \(\BWMP\) objective in \Cref{sec:expected-bounded-window-mean-payoff-value}.
We conclude in \Cref{sec:discussion}, where we remark on the differences between our algorithm and~\cite{CHH09}, look at the applicability of our algorithm to other objectives, and discuss practical implementations for our algorithm.

\section{Preliminaries}%
\label{sec:preliminaries}

\subparagraph*{Probability distributions.}
\AP
A \intro{probability distribution} over a finite non-empty set \(A\) is a function \(\Prob \colon A \to [0,1] \) such that \(\sum_{a \in A} \Prob(a) = 1\).  
We denote by \(\AP \intro *\DistributionSet{A}\) the set of all probability distributions over \(A\).
For the algorithmic and complexity results, we assume that probabilities are given as rational numbers.

\subparagraph*{Stochastic games.}
We consider two-player turn-based zero-sum stochastic games (or simply, stochastic games in the sequel). 
The two players are referred to as \(\PlayerMain\) (she/her) and \(\PlayerAdversary\) (he/him).
\AP 
A \intro{stochastic game} is given by \AP \(\Game = ((\Vertices, \Edges), (\VerticesMain, \VerticesAdversary, \VerticesRandom),  \ProbabilityFunction, \PayoffFunction)\), where:
\begin{itemize}
    \item \((\Vertices, \Edges)\) is a directed graph with 
    a finite set \(\Vertices\) of vertices and a set  \(\Edges \subseteq \Vertices \times \Vertices\) of directed edges such that
    for all vertices \(\Vertex \in \Vertices\), 
    the set \(\AP \intro *\OutNeighbours{\Vertex} = \{ \Vertex' \in \Vertices \suchthat (\Vertex, \Vertex') \in \Edges\}\) 
    of out-neighbours of \(\Vertex\) is non-empty, i.e., \(\OutNeighbours{\Vertex} \ne \emptyset\) (no deadlocks).
    \item \((\VerticesMain, \VerticesAdversary, \VerticesRandom)\) is a partition of \(\Vertices\). The vertices in \(\VerticesMain\) belong to \(\PlayerMain\), the vertices in~\(\VerticesAdversary\) belong to \(\PlayerAdversary\), and the vertices in \(\VerticesRandom\) are called \AP \intro{probabilistic vertices};
    \item \(\ProbabilityFunction \colon \VerticesRandom \to \DistributionSet{\Vertices}\) is a \emph{probability function} that describes the behaviour of \kl{probabilistic vertices} in the game. 
    It maps each \kl{probabilistic vertex} \(\Vertex \in \kl{\VerticesRandom}\) to a \kl{probability distribution} \(\ProbabilityFunction(\Vertex)\) over the set \(\OutNeighbours{\Vertex}\) of out-neighbours of \(\Vertex\) such that \(\ProbabilityFunction(\Vertex)(\Vertex') > 0\) for all \(\Vertex' \in \OutNeighbours{\Vertex}\) (i.e., all out-neighbours have non-zero probability);
    \item \(\PayoffFunction \colon \Edges \to \Rationals\) is a \intro{payoff function} assigning a rational payoff to every edge in the game.
\end{itemize}

\kl{Stochastic games} are played in rounds. 
The game starts by initially placing a token on some vertex.
At the beginning of a round, if the token is on a vertex \(\Vertex\), and \(\Vertex \in \Vertices[i]\) for \(i \in \{\Main, \Adversary\}\), then \(\Player{i}\) chooses an out-neighbour \(\Vertex' \in \OutNeighbours{\Vertex}\); otherwise \(\Vertex \in \kl{\VerticesRandom}\), and an out-neighbour \(v' \in \OutNeighbours{\Vertex}\) is chosen with probability \(\ProbabilityFunction(\Vertex)(v')\). 
\(\PlayerMain\) receives from \(\PlayerAdversary\) the amount \(\PayoffFunction(\Vertex, \Vertex')\) given by the \kl{payoff function}, and the token moves to \(v'\) for the next round. 
This continues ad infinitum resulting in an infinite sequence \(\Play = \Vertex[0] \Vertex[1] \Vertex[2] \dotsm \in \Vertices^\omega\) such that \((\Vertex[i], \Vertex[i+1]) \in \Edges\) for all \(i \ge 0\), called a \AP \intro{play}.
For \(i < j\), we denote by \(\PlayInfix{i}{j}\) the \emph{infix} \(\Vertex[i] \Vertex[i+1] \dotsm \Vertex[j]\) of \(\Play\). 
Its length is \( \abs{\PlayInfix{i}{j}} = j - i\), the number of edges. 
We denote by \(\PlayPrefix{j}\) the finite \emph{prefix} \(\Vertex[0] \Vertex[1] \dotsm \Vertex[j]\) of \(\Play\), and by \(\PlaySuffix{i}\) the infinite \emph{suffix} \(\Vertex[i] \Vertex[i+1] \ldots \) of \(\Play\). 
We denote by \(\occ(\Play)\) the set of vertices in \(\Vertices\) that occur at least once in \(\Play\), and by  \(\inf(\Play)\) the set of vertices in \(\Vertices\) that occur infinitely often in \(\Play\).
We denote by \(\PlaySet{\Game}\) and \(\PrefixSet{\Game}\) the set of all \kl{plays} and the set of all finite prefixes in \(\Game\) respectively.
We denote by \(\Last{\Prefix}\) the last vertex of the prefix \(\Prefix \in \PrefixSet{\Game}\). 
We denote by \(\PrefixSet[i]{\Game}\) (\(i \in \{\Main, \Adversary\}\)) the set of all prefixes \(\Prefix\) such that \(\Last{\Prefix} \in \Vertices[i]\).

A \kl{stochastic game} with \(\kl{\VerticesRandom} = \emptyset\) is a \AP \intro{non-stochastic two-player game},
a \kl{stochastic game} with \(\VerticesAdversary = \emptyset\) is a \AP \intro{Markov decision process (MDP)}, 
a \kl{stochastic game} with \(\VerticesMain = \VerticesAdversary = \emptyset\) is a \AP \intro{Markov chain},
and a \kl{stochastic game} with \(\VerticesMain = \VerticesRandom = \emptyset\) or \(\VerticesAdversary = \VerticesRandom = \emptyset\)  is called a non-stochastic \emph{one-player game}.
\Cref{fig:swmp-example} shows an example of a \kl{stochastic game}; \(\PlayerMain\) vertices are shown as circles, \(\PlayerAdversary\) vertices as boxes, and probabilistic vertices as diamonds. 

\begin{figure}[t]
    \centering
    \scalebox{0.92}{
    \begin{tikzpicture}[node distance=1.5cm]
        \node[square, draw] (v1) {\(v_1\)};
        
        \node[state, draw, right of=v1] (v3) {\(v_3\)};
        \node[random, draw, above of=v3] (v2) {\(v_2\)};
        \node[random, draw, right of=v3] (v5) {\(v_5\)};
        \node[square, draw, above of=v5] (v4) {\(v_4\)};
       
        \node[state, right of=v5] (v7) {\(v_7\)};
        \node[square, draw, above of=v7] (v6) {\(v_6\)};
        \node[random, draw, right of=v7] (v9) {\(v_9\)};
        \node[random, draw, above of=v9] (v8) {\(v_8\)};

        \node[state, right of=v9] (v11) {\(v_{11}\)};
        \node[square, draw, above of=v11] (v10) {\(v_{10}\)};
        
        \node[state, right of=v11] (v13) {\(v_{13}\)};
        \node[random, draw, above of=v13] (v12) {\(v_{12}\)};
        \node[square, draw, right of=v13] (v14) {\(v_{14}\)};

        \draw 
              (v1) edge[loop left] node[below, yshift=-1mm]{\small \(\EdgeValues{-2}{}\)} (v1)
        ;
        
        \draw
              (v3) edge[left, pos=0.3] node{\small \(\EdgeValues{2}{}\)} (v2)
              (v3) edge[below, pos=0.3] node{\small \(\EdgeValues{7}{}\)} (v5)
              
              (v4) edge[above left, pos=0.3] node[xshift=1mm, yshift=-0.5mm]{\(\EdgeValues{}{}\)} (v5)
        ;
        \draw[transform canvas={xshift=-0.4mm, yshift=+0.4mm}]
              (v3) edge[above left, pos=0.2] node[xshift=1.2mm, yshift=-0.2mm]{\small \(\EdgeValues{2}{}\)} (v4)
        ;
        \draw[transform canvas={xshift=+0.4mm, yshift=-0.4mm}]
              (v4) edge[below right, pos=0.2] node[xshift=-1.2mm, yshift=0.2mm]{\small \(\EdgeValues{0}{}\)} (v3)
        ;

        \draw
              (v6) edge[loop left, pos=0.3] node[xshift=2mm, yshift=-2mm]{\small \(\EdgeValues{0}{}\)} (v6)
              (v6) edge[above, pos=0.3] node[xshift=0.5mm, yshift=-0.5mm]{\small \(\EdgeValues{1}{}\)} (v8)

              (v7) edge[left, pos=0.3] node[xshift=0.5mm, yshift=0mm]{\small \(\EdgeValues{-3}{}\)} (v6)
              (v7) edge[below, pos=0.3] node[xshift=1mm, yshift=0mm]{\small \(\EdgeValues{-6}{}\)} (v9)
             
              (v9) edge[left, pos=0.3] node[xshift=1.4mm, yshift=-1.5mm]{\small \(\EdgeValues{0}{.9}\)} (v6)
              (v9) edge[right, pos=0.3] node[xshift=-0.2mm, yshift=-1mm]{\small \(\EdgeValues{5}{.1}\)} (v8)
        ;
        
        \draw[transform canvas={xshift=+0.6mm}]
              (v10) edge[right, pos=0.3] node{\(\EdgeValues{2}{}\)} (v11)
        ;
        \draw[transform canvas={xshift=-0.6mm}]
              (v11) edge[left, pos=0.3] node{\(\EdgeValues{0}{}\)} (v10)
        ;

        \draw
              (v12) edge[loop right, pos=0.3] node[xshift=1mm]{\small \(\EdgeValues{9}{.3}\)} (v12)
              (v12) edge[auto, pos=0.3] node[xshift=-0.5mm]{\small \(\EdgeValues{-6}{.7}\)} (v13)

              (v13) edge[below, pos=0.3] node{\small \(\EdgeValues{-4}{}\)} (v14)

              (v14) edge[loop right] node[below, xshift=-0.2mm, yshift=-1.2mm]{\small \(\EdgeValues{2}{}\)}(v14)
        ;

        \draw
              (v2) edge[above left, pos=0.3] node[xshift=1mm, yshift=-1mm]{\small \(\EdgeValues{1}{.5}\)} (v1)
              (v2) edge[bend left, above, pos=0.2] node[yshift=-0.5mm]{\small \(\EdgeValues{4}{.5}\)} (v6)

              (v5) edge[bend left, below, pos=0.3] node{\small \(\EdgeValues{0}{.75}\)} (v1)
              (v5) edge[bend right=25, below, pos=0.15] node[xshift=-2.5mm]{\small \(\EdgeValues{-1}{.25}\)} (v13)

              (v8) edge[above, pos=0.3] node{\small \(\EdgeValues{2}{.5}\)} (v10)
              (v8) edge[bend right, above, pos=0.3] node{\small \(\EdgeValues{1}{.5}\)} (v4)
              
              (v1) edge[above, pos=0.3] node[xshift=0.25mm, yshift=-0.5mm]{\small \(\EdgeValues{1}{}\)} (v3)
              
              (v7) edge[below, pos=0.3] node[yshift=0.5mm]{\small \(\EdgeValues{9}{}\)} (v5)
              
              (v11) edge[below, pos=0.3] node[xshift=0mm, yshift=0mm]{\small \(\EdgeValues{1}{}\)} (v9)
              
              (v10) edge[above, pos=0.3] node[xshift=0mm, yshift=0mm]{\small \(\EdgeValues{1}{}\)} (v12)
              (v13) edge[below, pos=0.3] node[xshift=-1.7mm, yshift=0.2mm]{\small \(\EdgeValues{1}{}\)} (v11)
        ;
    \end{tikzpicture}
    }
    \caption{A \kl{stochastic game}. 
    \(\PlayerMain\) vertices are denoted by circles, \(\PlayerAdversary\) vertices are denoted by boxes, and probabilistic vertices are denoted by diamonds.
    The payoff for each edge is shown in red and probability distribution out of probabilistic vertices is shown in blue.}
    \label{fig:swmp-example}
\end{figure}

\subparagraph*{Subgames and traps.}
Given a \kl{stochastic game} \(\Game = ((\Vertices, \Edges), (\VerticesMain, \VerticesAdversary, \VerticesRandom),  \ProbabilityFunction, \PayoffFunction)\), a subset \(\Vertices' \subseteq \Vertices\) of vertices \emph{induces} a subgame if $(i)$
every vertex  \(\Vertex' \in \Vertices'\) has an outgoing edge in $\Vertices'$, that is
\(\OutNeighbours{\Vertex'} \cap \Vertices' \neq \emptyset\), and $(ii)$
every probabilistic vertex \(\Vertex' \in \VerticesRandom \intersection \Vertices'\) has all 
outgoing edges in~$\Vertices'$, that is \(\OutNeighbours{\Vertex'} \subseteq \Vertices'\).
The induced \emph{subgame} is 
\(((\Vertices',  \Edges'),
(\VerticesMain \intersection \Vertices', \VerticesAdversary \intersection \Vertices', \VerticesRandom \intersection \Vertices'), \ProbabilityFunction',
\PayoffFunction')\), where \(\Edges' = \Edges \intersection (\Vertices' \times \Vertices')\), and \(\ProbabilityFunction'\) and \(\PayoffFunction'\) are restrictions of \(\ProbabilityFunction\) and \(\PayoffFunction\) respectively to \((\Vertices', \Edges')\).
Let \(\Objective\) be an objective in the \kl{stochastic game} \(\Game\). We define the restriction of \(\Objective\) to a subgame~\(\Game'\) of \(\Game\) 
to be the set of all \kl{plays} in \(\Game'\) satisfying \(\Objective\), that is, the set \(\PlaySet{\Game'} \intersection \Objective\).
If \(\Game\) is an MDP, then a subgame \(\Game'\) of \(\Game\) is also an MDP, and is called a \emph{subMDP} of \(\Game\).

If \(T \subseteq \Vertices\) is such that for all \(v \in T\), 
if \(v \in \VerticesMain \union \VerticesRandom\) then \(\OutNeighbours{v} \subseteq T\)
and if \(v \in \VerticesAdversary\) then \(\OutNeighbours{v} \intersection T \ne \emptyset\), then \(T\) induces a subgame, and the subgame is a \intro{trap} for \(\PlayerMain\) in \(\Game\), since \(\PlayerAdversary\) can ensure that if the token reaches \(T\), then it never escapes.

\subparagraph*{Boolean objectives.}
\AP
A \intro{Boolean objective} \(\BooleanObjective\) is a Borel-measurable subset of \(\PlaySet{\Game}\)~\cite{Mar98}.
A \kl{play} \(\Play \in \PlaySet{\Game}\) \emph{satisfies} an objective \(\BooleanObjective\) if \(\Play \in \BooleanObjective\). 
In a \kl{stochastic game} \(\Game\) with objective \(\BooleanObjective\), the objective of \(\PlayerMain\) is \(\BooleanObjective\), and since \(\Game\) is a zero-sum game, the objective of \(\PlayerAdversary\) is the complement set \(\ComplementObjective{\BooleanObjective} = \PlaySet{\Game} \setminus \BooleanObjective\). 
Given \(T \subseteq \Vertices\), define the following Boolean objectives:
\begin{itemize}
\item the \intro{reachability} objective \(\kl{\ReachObj}_{\Game}(T) \define \{ \Play \in \PlaySet{\Game} \suchthat T \intersection \occ(\Play) \ne \emptyset \}\), the set of all \kl{plays} that visit a vertex in \(T\),
\item the dual \emph{safety} objective \(\SafeObj_{\Game}(T) \define \{ \Play \in \PlaySet{\Game} \suchthat \occ(\Play) \subseteq T \}\),
the set of all \kl{plays} that never visit a vertex outside \(T\),
\item 
the \emph{\Buchi} objective \(\textsf{\Buchi}_{\Game}(T) \define \{ \Play \in \PlaySet{\Game} \suchthat T \intersection \inf(\Play) \ne \emptyset \}\), the set of all \kl{plays} that visit a vertex in~\(T\) infinitely often, and
\item 
the dual \emph{co\Buchi} objective \(\textsf{co\Buchi}_{\Game}(T) \define \{ \Play \in \PlaySet{\Game} \suchthat \inf(\Play) \subseteq T \}\), the set of all \kl{plays} that eventually only visit vertices in \(T\). 
\end{itemize}

\subparagraph*{Quantitative objectives.}
\AP
A \intro{quantitative objective} is a Borel-measurable function of the form \(\Objective \colon \PlaySet{\Game} \to \Reals\). 
In a \kl{stochastic game} \(\Game\) with objective \(\Objective\), the objective of \(\PlayerMain\) is \(\Objective\) and the objective of \(\PlayerAdversary\) is \(-\Objective\), the negative of \(\Objective\). 
Let \(\Play = v_0 v_1 v_2 \cdots\) be a \kl{play}. 
Some common examples of \kl{quantitative objectives} include 
the \emph{mean-payoff} objective
\(\MP(\Play) = \liminf_{n \to \infty} \frac{1}{n} \sum_{i=0}^{n} \PayoffFunction(v_i, v_{i+1}) \), and
the \emph{liminf} objective
\(\liminfObj(\Play) = \liminf_{n \to \infty} \PayoffFunction(v_n, v_{n+1})\).
In this work, we also consider the \emph{window mean-payoff} objective, which is defined in \Cref{sec:window-mean-payoff-objectives}.
Corresponding to a \kl{quantitative objective} \(\Objective\), we define \AP \intro{threshold objectives} which are \kl{Boolean objectives} \(\BooleanObjective\) of the form \(\{\Play \in \PlaySet{\Game} \suchthat \Objective(\Play) \bowtie \Threshold\}\) for thresholds \(\Threshold \in \Reals\) and for \(\bowtie ~\in \{<, \le, >, \ge\}\).
We denote this \kl{threshold objective} succinctly as \(\ThresholdObjective{\bowtie \Threshold}\).

\subparagraph*{Prefix independence.}
\AP An objective is said to be \intro{prefix-independent} if it only depends on the suffix of a \kl{play}.
Formally, a \kl{Boolean objective} \(\BooleanObjective\) is \kl{prefix-independent} if for all \kl{plays} \(\Play\) and \(\Play'\) with a common suffix (that is, \(\Play'\) can be obtained from \(\Play\) by removing and adding a finite prefix), we have that \(\Play \in \BooleanObjective\) if and only if \(\Play' \in \BooleanObjective\). 
Similarly, a \kl{quantitative objective} \(\Objective\) is \kl{prefix-independent} if for all \kl{plays} \(\Play\) and \(\Play'\) with a common suffix, we have that \(\Objective(\Play) = \Objective(\Play')\). 
Mean payoff and liminf are examples of \kl{prefix-independent} objectives, whereas \kl{reachability} and discounted payoff~\cite{AG11} are not.

\subparagraph*{Strategies.}
A (deterministic or pure) \emph{strategy}\footnote{
We only consider the satisfaction and expectation of Borel-measurable objectives, and deterministic strategies suffice for such objectives~\cite{CDGH15}.
Satisfying two goals simultaneously, e.g., \(\Prob(\kl{\ReachObj}(T_1)) > 0.5 \wedge \Prob(\kl{\ReachObj}(T_2)) > 0.5\) requires randomization and is not allowed by our definition.
} for Player \(\i \in \{\Main, \Adversary\}\) in a game \(\Game\) is a function \(\Strategy[i]: \PrefixSet[i]{\Game} \to \Vertices\) that maps prefixes ending in a vertex \(\Vertex \in \Vertices[i]\) to a successor of \(\Vertex\). 
Strategies can be realized as the output of a (possibly infinite-state) Mealy machine~\cite{HU79}.
A \emph{Mealy machine} is a deterministic transition system with transitions labelled by input/output pairs. 
Formally, a Mealy machine \(M\) is a tuple \((Q, q_0, \Sigma_i, \Sigma_o, \Delta, \delta)\) 
where 
\(Q\) is the set of states of \(M\) (the memory of the induced strategy), \(q_0 \in Q\) is the initial state,
\(\Sigma_i\) is the input alphabet, 
\(\Sigma_o\) is the output alphabet,
\(\Delta \colon Q \times \Sigma_i \to Q\) is a transition function that reads the current state of~\(M\) and an input letter and returns the next state of  \(M\), and \(\delta \colon Q \times \Sigma_i \to \Sigma_o \) is an output function that reads the current state of
\(M\) and an input letter and returns an output letter.
The transition function \(\Delta\) can be extended to a function \(\hat{\Delta} \colon Q \times \Sigma_i^+ \to Q\) that reads words and can be defined inductively by \(\hat{\Delta}(q, a) = \Delta(q, a)\) and \(\hat{\Delta}(q, x \cdot a) = \Delta(\hat{\Delta}(q, x), a) \), for \(q \in Q\), \(x \in \Sigma_i^+\), and \(a \in \Sigma_i\).
The output function~\(\delta\) can be also be similarly extended to a function  \(\hat{\delta} \colon Q \times \Sigma_i^+ \to \Sigma_o\) on words and can be defined inductively by \(\hat{\delta}(q, a) = \delta(q, a)\) and \(\hat{\delta}(q, x \cdot a) = \delta(\hat{\Delta}(q, x), a) \), for \(q \in Q\), \(x \in \Sigma_i^+\), and \(a \in \Sigma_i\).

A player's strategy can be defined by a Mealy machine whose input and output alphabets are \(\Vertices\) and \(\Vertices \union \{\epsilon\}\) respectively. 
For \(i \in \{\Main, \Adversary\}\), a strategy \(\Strategy[i]\) of \(\PlayerI\) can be defined by a Mealy machine \((Q, q_0, V, V \union \{\epsilon\}, \Delta, \delta)\) as follows: 
Given a prefix \(\Prefix \in \PrefixSet[i]{\Game}\) ending in a \(\PlayerIDash\) vertex, the strategy \(\Strategy[i]\) defined by a Mealy machine is \(\Strategy[i](\Prefix) = \hat{\delta}(q_0, \Prefix)\).
Intuitively, in each turn, if the token is on a vertex~\(v\) that belongs to \(\PlayerI\) for \(i \in \{\Main, \Adversary\}\), then~\(v\) is given as input to the Mealy machine, and the Mealy machine outputs the successor vertex of~\(v\) that \(\PlayerI\) must choose. 
Otherwise, the token is on a vertex \(v\) that either belongs to \(\PlayerI\)'s opponent or is a probabilistic vertex, in which case, the Mealy machine outputs the symbol \(\epsilon\) to denote that \(\PlayerI\) cannot decide the successor vertex of \(v\).
Memoryless strategies can be defined by Mealy machines with only one state.

The \emph{memory size} of a strategy \(\Strategy[i]\) is the smallest number of states a Mealy machine defining \(\Strategy[i]\) can have. 
A strategy $\Strategy[\i]$ is \emph{memoryless} if $\Strategy[\i](\Prefix)$ only depends on the last element of the prefix~\(\Prefix\), that is, for all prefixes \(\Prefix, \Prefix' \in \PrefixSet[\i]{\Game}\) if \(\Last{\Prefix} = \Last{\Prefix'}\), then \(\Strategy[\i](\Prefix) = \Strategy[\i](\Prefix')\).

A \emph{strategy profile} \(\StrategyProfile = (\Strategy[\Main], \Strategy[\Adversary])\) is a pair of strategies \(\Strategy[\Main]\) and \(\Strategy[\Adversary]\) of \(\PlayerMain\) and \(\PlayerAdversary\) respectively.
A \kl{play} \(\Play = \Vertex[0] \Vertex[1] \dotsm\) is \emph{consistent} with a strategy \(\Strategy[i] \) of \(\PlayerI\) (\(i \in \{\Main, \Adversary\}\)) if for all \(j \geq 0\) with \(\Vertex[j] \in \Vertices[i]\), we have \(\Vertex[j+1] = \Strategy[i](\PlayPrefix{j})\). 
A \kl{play} \(\Play\) is an \emph{outcome} of a profile \(\StrategyProfile = (\Strategy[\Main], \Strategy[\Adversary])\) if it is consistent with both \(\Strategy[\Main]\) and \(\Strategy[\Adversary]\). 
For a \kl{Boolean objective} \(\BooleanObjective\), we denote by \(\Pr{\Strategy[\Main], \Strategy[\Adversary]}{\Game, \Vertex}{\BooleanObjective}\) the probability that an outcome of the profile \((\Strategy[\Main], \Strategy[\Adversary])\) in \(\Game\) with initial vertex \(\Vertex\) satisfies \(\BooleanObjective\). 
Similarly, we denote by \(\Pr{\Strategy[\Main], \Strategy[\Adversary]}{\Game, \Prefix}{\BooleanObjective}\) the probability that an outcome of the profile \((\Strategy[\Main], \Strategy[\Adversary])\) in \(\Game\) with initial prefix \(\Prefix\) satisfies \(\BooleanObjective\). 

The \emph{cone} at \(\Prefix\) is the set \(\Cone{\Prefix} \define \{\Play \in \PlaySet{\Game} \suchthat \Prefix \text{ is a prefix of } \Play \}\), the set of all \kl{plays} having \(\Prefix\) as a prefix. 
First, we define this probability measure over cones inductively as follows. 
If \(\abs{\Prefix} = 0\), then \(\Prefix \) is just a vertex \(\Vertex[0]\), and 
\(\Pr{\Strategy[\Main], \Strategy[\Adversary]}{\Game, \Vertex}{\Cone{\Prefix}}\) is $1$ if $\Vertex = \Vertex[0]$, and $0$ otherwise. 
For the inductive case \(\abs{\Prefix} > 0\), 
there exist \(\Prefix' \in \PrefixSet{\Game}\) and \(\Vertex' \in \Vertices\) such that \(\Prefix = \Prefix' \cdot \Vertex'\), and we have
\(\Pr{\Strategy[\Main], \Strategy[\Adversary]}{\Game, \Vertex}{\Cone{\Prefix' \cdot \Vertex'}} \) given by the following:
\begin{equation*}
\begin{cases}
\Pr{\Strategy[\Main], \Strategy[\Adversary]}{\Game, \Vertex}{\Cone{\Prefix'}}  \cdot \ProbabilityFunction(\Last{\Prefix'})(\Vertex') & \text{if } \Last{\Prefix'} \in \VerticesRandom,\\
\Pr{\Strategy[\Main], \Strategy[\Adversary]}{\Game, \Vertex}{\Cone{\Prefix'}}   
& \text{if } \Last{\Prefix'} \in \Vertices[i]\\
& \text{ and } \Strategy[i](\Prefix') = \Vertex', \\
0 & \text{otherwise.}\\
\end{cases}
\end{equation*}
It is sufficient to define \(\Pr{\Strategy[\Main], \Strategy[\Adversary]}{\Game, \Vertex}{\BooleanObjective}\) on cones in \(\Game\) since a measure defined on cones extends to a unique measure on \(\PlaySet{\Game}\) by Carath\'{e}odory's extension theorem~\cite{Billingsley86}.

\subparagraph*{Satisfaction probability of Boolean objectives.}
Let \(\BooleanObjective\) be a \kl{Boolean objective}.
A strategy \(\Strategy[\Main]\) of \(\PlayerMain\)  is \emph{winning} with probability \(p\) from a vertex \(\Vertex\) in \(\Game\) for objective \(\BooleanObjective\) if \(\Pr{\Strategy[\Main], \Strategy[\Adversary]}{\Game, \Vertex}{\BooleanObjective} \ge p\) for all strategies \(\Strategy[\Adversary]\) of \(\PlayerAdversary\).
A strategy \(\Strategy[\Main]\) of \(\PlayerMain\) is \emph{positive} winning (resp., \emph{almost-sure} winning) from \(v\) for \(\PlayerMain\) in~\(\Game\) with objective \(\BooleanObjective\) if 
\(\Pr{\Strategy[\Main], \Strategy[\Adversary]}{\Game, \Vertex}{\BooleanObjective} > 0\) (resp., \(\Pr{\Strategy[\Main], \Strategy[\Adversary]}{\Game, \Vertex}{\BooleanObjective} = 1\)) for all strategies \(\Strategy[\Adversary]\) of \(\PlayerAdversary\).
In the above, if such a strategy \(\Strategy[\Main]\) exists, then the vertex \(v\) is said to be positive winning (resp., almost-sure winning) for \(\PlayerMain\).
If a vertex $v$ is positive winning (resp., almost-sure winning) for \(\PlayerMain\), 
then \(\PlayerMain\) is said to \kl{play} \emph{optimally} from $v$ if she follows a positive (resp., almost-sure) winning strategy from $v$.
We omit analogous definitions for \(\PlayerAdversary\).

\subparagraph{Maximal end components (MECs).}
Let \(\Game = ((\Vertices, \Edges), (\VerticesMain, \emptyset, \VerticesRandom), \ProbabilityFunction, \PayoffFunction)\) be an MDP. 
Given a subset \(\Vertices' \subseteq \Vertices\) of vertices, the subMDP \(\Game'\) of \(\Game\) induced by \(\Vertices'\) is an \emph{end component} of \(\Game\) if there exists a strategy \(\Strategy[\Main]\) such that for all vertices \(v\), \(v'\) in the subMDP \(\Game'\), we have that \(v'\) is almost surely reachable from \(v\) in \(\Game'\).
Intuitively, if the token is in an end component \(\Game'\), then \(\PlayerMain\) has a strategy such that, with probability~\(1\), the token never leaves the end component and every vertex is reachable from every other vertex in the end component. 
An end component that is not contained in any other end component is called a \emph{maximal end component} (MEC). 

MECs have traditionally been defined for MDPs only.
Recent works such as~\cite{AKW19} have extended the definition of MECs to \kl{stochastic games} in general.
Let \(\Game = ((\Vertices, \Edges), (\VerticesMain, \VerticesAdversary, \VerticesRandom), \ProbabilityFunction, \PayoffFunction)\) be a \kl{stochastic game}. 
Given a subset \(\Vertices' \subseteq \Vertices\) of vertices, the subgame \(\Game'\) of \(\Game\) induced by \(\Vertices'\) is an \emph{end component} of \(\Game\) if there exists a strategy profile \((\Strategy[\Main], \Strategy[\Adversary])\) such that for all vertices \(v\), \(v'\) in the subgame \(\Game'\), we have that \(v'\) is almost surely reachable from \(v\) in \(\Game'\).
In other words, if the players \emph{cooperate}, then they can ensure that once the token enters an end component, then it never leaves the end component, and that every vertex is reachable from every other vertex in the end component. 
Equivalently, the subgame \(\Game'\) of \(\Game\) induced by \(\Vertices'\) is an end component of \(\Game\) if, in the MDP \(\Game_{\mathsf{MDP}} = ((\Vertices, \Edges), (\VerticesMain \union \VerticesAdversary, \emptyset, \VerticesRandom), \ProbabilityFunction, \PayoffFunction)\) obtained by replacing all \(\PlayerAdversary\) vertices in \(\Game\) by \(\PlayerMain\) vertices, the subMDP \(\Game_{\mathsf{MDP}}'\) of \(\Game_{\mathsf{MDP}}\) induced by \(\Vertices'\) is an end component of \(\Game_{\mathsf{MDP}}\) (in the sense of an end component of an MDP).

\subparagraph*{Expected value of quantitative objectives.}
Let \(\Objective\) be a \kl{quantitative objective}. 
Given a strategy profile \(\StrategyProfile = (\Strategy[\Main], \Strategy[\Adversary])\) and an initial vertex \(v\), let \(\ExpOutcome{v}{\Objective}{\StrategyProfile}\) denote the \emph{expected \(\Objective\)-value of the outcome} of the strategy profile \(\StrategyProfile\) from \(v\), that is,
the expectation of \(\Objective\) over all \kl{plays} with initial vertex \(v\) under the probability measure \(\Pr{\Strategy[\Main], \Strategy[\Adversary]}{\Game, v}{\Objective}\).
Similarly, we also define \(\ExpOutcome{\Prefix}{\Objective}{\StrategyProfile}\), the expected \(\Objective\)-value of an outcome of the strategy profile \(\StrategyProfile\) following a prefix \(\Prefix\).

We only consider objectives \(\Objective\) that are Borel-measurable and whose image is bounded.
Thus, by determinacy of Blackwell games~\cite{Mar98}, we have that \kl{stochastic games} with objective \(\Objective\) are determined.
That is, we have
\(
    \sup_{\Strategy[\Main]} \inf_{\Strategy[\Adversary]} \ExpOutcome{v}{\Objective}{\StrategyProfile} = \inf_{\Strategy[\Adversary]} \sup_{\Strategy[\Main]} \ExpOutcome{v}{\Objective}{\StrategyProfile}
\).
We call this quantity the \emph{expected \(\Objective\)-value of the vertex} \(v\) and denote it by \(\ExpValue{v}{\Objective}\).
We say that \(\PlayerMain\) plays \emph{optimally} from a vertex \(v\) if she follows a strategy \(\Strategy[\Main]\) such that for all strategies \(\Strategy[\Adversary]\) of \(\PlayerAdversary\), the expected \(\Objective\)-value of the outcome is at least \(\ExpValue{v}{\Objective}\).
Similarly, \(\PlayerAdversary\) plays \emph{optimally} if he follows a strategy \(\Strategy[\Adversary]\) such that for all strategies \(\Strategy[\Main]\) of \(\PlayerMain\), the expected \(\Objective\)-value of the outcome is at most \(\ExpValue{v}{\Objective}\).
If \(\Objective\) is a \kl{prefix-independent} objective, then we have the following relation between the expected \(\Objective\)-value of a vertex \(v\) and the expected \(\Objective\)-values of its out-neighbours.

\begin{restatable}[\intro{Bellman equations}]{proposition}{bellman}%
\label{prop:expected-wmp-value-out-neighbours}
    If \(\Objective\) is a \kl{prefix-independent} objective, then the following equations hold for all \(v \in \Vertices\).
    \[
        \ExpValue{v}{\Objective} = 
        \begin{cases}
            \max_{v' \in \OutNeighbours{v}} \ExpValue{v'}{\Objective}
            & \text{if } v \in \VerticesMain \\
            \min_{v' \in \OutNeighbours{v}} \ExpValue{v'}{\Objective} 
            & \text{if } v \in \VerticesAdversary  \\
            \sum_{v' \in \OutNeighbours{v}} \ProbabilityFunction(v)(v') \cdot \ExpValue{v'}{\Objective} & \text{if } v \in \VerticesRandom
        \end{cases}
    \]
\end{restatable}
\begin{proof}
    First, we show these equations hold for every probabilistic vertex \(v \in \VerticesRandom\).
    For every strategy profile \(\StrategyProfile\), we have that 
    \(\ExpOutcome{v}{\Objective}{\StrategyProfile} = \sum_{v' \in \OutNeighbours{v}} \ProbabilityFunction(v)(v') \cdot \ExpOutcome{v \cdot v'}{\Objective}{\StrategyProfile}\), that is, the expected \(\Objective\)-value of an outcome of the profile \(\StrategyProfile\) with initial vertex \(v\) is equal to the \(\ProbabilityFunction(v)\)-weighted average of the expected \(\Objective\)-value of the outcomes of \(\StrategyProfile\) with prefix \(v \cdot v'\) for out-neighbours \(v'\) of \(v\).
    Since \(\Objective\) is \kl{prefix-independent}, the \(\Objective\)-value of a \kl{play} is not affected by removing the prefix \(v\cdot v'\).
    Thus, we have that \(\ExpOutcome{v \cdot v'}{\Objective}{\StrategyProfile} = \ExpOutcome{v'}{\Objective}{\StrategyProfile}\) for all out-neighbours \(v'\) of \(v\).
    Thus, for all strategy profiles \(\StrategyProfile\), we see that 
    \(\ExpOutcome{v}{\Objective}{\StrategyProfile} = \sum_{v' \in \OutNeighbours{v}} \ProbabilityFunction(v)(v') \cdot \ExpOutcome{v'}{\Objective}{\StrategyProfile}\).
    In particular, after taking the supremum and the infimum over the strategies of the players, we get
    \(\ExpValue{v}{\Objective} = \sum_{v' \in \OutNeighbours{v}} \ProbabilityFunction(v)(v') \cdot \ExpValue{v'}{\Objective}\).

    Now, we show that \(\ExpValue{v}{\Objective} = \max_{v'\in \OutNeighbours{v}}\ExpValue{v'}{\Objective}\) for every \(\PlayerMain\) vertex \(v \in \VerticesMain\).
    First, we show that \(\ExpValue{v}{\Objective} \ge \max_{v'\in \OutNeighbours{v}}\ExpValue{v'}{\Objective}\).
    Suppose \(u \in \OutNeighbours{v}\) is an out-neighbour of \(v\) with the maximum expected \(\Objective\)-value, that is, \(\ExpValue{u}{\Objective} \ge \ExpValue{v'}{\Objective}\) for all \(v' \in \OutNeighbours{v}\).
    Since \(\Objective\) is \kl{prefix-independent}, for any expected \(\Objective\)-value that can be achieved from \(u\), the same expected \(\Objective\)-value can also be achieved from \(v\) by first moving the token to \(u\), and then playing as if the game started from \(u\).
    Thus, we get that \(\ExpValue{v}{\Objective} \ge \max_{v'\in \OutNeighbours{v}}\ExpValue{v'}{\Objective}\).
    Now, we show the other direction \(\ExpValue{v}{\Objective} \le \max_{v'\in \OutNeighbours{v}}\ExpValue{v'}{\Objective}\).
    Starting from \(v\), in the most general setting, \(\PlayerMain\) chooses an out-neighbour according to some probability distribution.
    Since \(\Objective\) is \kl{prefix-independent}, the expected \(\Objective\)-value of \(v\) is a convex combination of the expected \(\Objective\)-values of out-neighbours of \(v\), which is at most \(\max_{v'\in \OutNeighbours{v}}\ExpValue{v'}{\Objective}\).
    
    We omit the proof for \(\PlayerAdversary\) vertices \(v \in \VerticesAdversary\) as it is analogous to the case of \(\PlayerMain\) vertices.
\end{proof}

In this paper, we consider the \kl{expectation problem} for \kl{prefix-independent} objectives.
Our solution in turn uses the \kl{almost-sure satisfaction problem}.
The decision problems are defined as follows.

\subparagraph*{Decision problems.}
Given a \kl{stochastic game} \(\Game\), a \kl{quantitative objective} \(\Objective\), a vertex \(v\), and a threshold \(\Threshold \in \Rationals\), the following decision problems are relevant:
\begin{itemize}
    \item  \AP \intro{almost-sure satisfaction problem}: 
    Is vertex \(v\) almost-sure winning for \(\PlayerMain\) for a \kl{threshold objective} \(\ThresholdObjective{> \Threshold}\)?
    \item \AP \intro{expectation problem}:
    Is \(\ExpValue{v}{\Objective} \ge \Threshold\)? That is, is the expected \(\Objective\)-value of \(v\) at least \(\Threshold\)? 
\end{itemize}

The reader is pointed to~\cite{AG11} and~\cite{FBB24} for a more comprehensive discussion on the above-mentioned concepts.

\section{Reducing expectation to almost-sure satisfaction}%
\label{sec:reducing-expectation-to-almost-sure-satisfaction}

In this section, we show a reduction (\Cref{thm:six-conditions}) of the \kl{expectation problem} for bounded quantitative \kl{prefix-independent} objectives \(\Objective\) to the \kl{almost-sure satisfaction problem} for the corresponding \kl{threshold objectives} \(\ThresholdObjective{> \Threshold}\).
The reduction involves guessing a value \(\Val{v}\) for every vertex \(v\) in the game, and then verifying if the guessed values are equal to the expected \(\Objective\)-values of the vertices.
\Cref{thm:six-conditions} generalizes~\cite[Lemma~8]{CHH09} which studies the satisfaction problem for \kl{prefix-independent} \kl{Boolean objectives}, as \kl{Boolean objectives} can be viewed as a special case of \kl{quantitative objectives} by restricting the range to \(\{0, 1\}\).
We further discuss the difference in approaches between \Cref{thm:six-conditions} and \cite[Lemma~8]{CHH09} in \Cref{sec:discussion}.

Given a game \(\Game\) and a bounded \kl{prefix-independent} \kl{quantitative objective} \(\Objective\), 
our reduction requires the existence of an integer bound \(\DenBound{\Objective}\) on the denominators of expected \(\Objective\)-values of vertices in \(\Game\).
Since \(\Objective\) is bounded, there exists an integer \(\ObjectiveBound{\Objective}\) such that  \(\abs{\Objective(\Play)} \le \ObjectiveBound{\Objective}\) for every \kl{play} \(\Play\) in \(\Game\).
Thus, for every vertex \(v\) in \(\Game\), one can write \(\ExpValue{v}{\Objective}\) as \(\frac{p}{q}\), where \(p\) and \(q\) are integers such that \(\abs{p} \le \ObjectiveBound{\Objective} \cdot \DenBound{\Objective}\) and \(0 < q \le \DenBound{\Objective} \).
The bounds \(\ObjectiveBound{\Objective}\) and \(\DenBound{\Objective}\) may depend on the objective and the structure of the graph, i.e., number of vertices, edge payoffs, probability distributions in the game, etc.
These bounds effectively discretize the set of possible expected \(\Objective\)-values of the vertices, as there are at most \((2 \cdot  \ObjectiveBound{\Objective} \cdot \DenBound{\Objective} + 1) \cdot \DenBound{\Objective}\) distinct possible values.
This directly gives a bound on the granularity of the possible expected \(\Objective\)-values of vertices, that is, the minimum difference between two possible values of vertices, and we represent this quantity by \(\GranularityBound{\Objective}\).
Observe that given two rational numbers with denominators at most \(\DenBound{\Objective}\), the difference between them is at least \((\frac{1}{\DenBound{\Objective}})^2\), and thus, we let \(\GranularityBound{\Objective}\) be \((\frac{1}{\DenBound{\Objective}})^2\).

\subsection{Value vectors and value classes}
We first define and give notations for \kl{value vectors}, which are useful in describing the reduction, and then look at some of their interesting and useful properties.
\subparagraph*{Definitions and notations.}
A \AP \intro(value){vector} \(\ValVector = (\Val{v})_{v \in \Vertices}\) of reals indexed by vertices in \(\Vertices\) induces a partition  of \(\Vertices\) such that all vertices with the same value in \(\ValVector\) belong to the same part in the partition.
Let \(\ValClassCount\) denote the number of parts in the partition, and let us denote the parts by \(\{\ValClass{1}, \ValClass{2}, \ldots, \ValClass{\ValClassCount}\} \).
We call each part \(\ValClass{i}\) of the partition an \(\ValVector\)-\emph{class}, or simply, \emph{class} if \(\ValVector\) is clear from the context.
For all \(1 \le i \le \ValClassCount\), let \(\ClassVal{i}\) denote the \(\ValVector\)-value of the class \(\ValClass{i}\). 
Given two vectors \(\ValVector, \ValVectorS\), we write \(\ValVector \ge \ValVectorS\) if for all \(v \in \Vertices\), we have \(\mathsf{r}_{v} \ge \mathsf{s}_{v}\), and we write \(\ValVector > \ValVectorS\) if we have that \(\ValVector \ge \ValVectorS\) and there exists \(v \in \Vertices\) such that \(\mathsf{r}_{v} > \mathsf{s}_{v}\).
For a constant \(c \in \Reals\), we denote by \(\ValVector + c\) the vector obtained by adding \(c\) to each component of \(\ValVector\).

For all \(1 \le i \le \ValClassCount\), a vertex \(v \in \ValClass{i}\) is a \AP \intro{boundary vertex} if \(v\) is a \kl{probabilistic vertex} and has an out-neighbour not in \(\ValClass{i}\), i.e., if \(v \in \kl{\VerticesRandom}\) and \(\OutNeighbours{v} \not \subseteq \ValClass{i}\). 
Let \(\Bnd{i}\) denote the set of \kl{boundary vertices} in the class \(\ValClass{i}\).
For all \(1 \le i \le \ValClassCount\), let \(\Game_{\ValClass{i}}\) denote the restriction of \(\Game\) to vertices in \(\ValClass{i}\) with all vertices in \(\Bnd{i}\) changed to absorbing vertices with a self-loop.
The edge payoffs of these self loops are not important (we assume them to be \(0\)) as we restrict our attention to a subgame of \(\Game_{\ValClass{i}}\) that does not contain \kl{boundary vertices}.
\begin{example}%
\label{ex:vc-notations}
    \begin{figure}[t]
        \centering
        \scalebox{0.92}{
        \begin{tikzpicture}[node distance=1.5cm]
            \node[square, draw] (v1) {\(v_1\)};
            
            \node[state, draw, right of=v1, xshift=4mm] (v3) {\(v_3\)};
            \node[random, draw, above of=v3] (v2) {\(v_2\)};
            \node[random, draw, right of=v3] (v5) {\(v_5\)};
            \node[square, draw, above of=v5] (v4) {\(v_4\)};
           
            \node[state, right of=v5, xshift=9mm] (v7) {\(v_7\)};
            \node[square, draw, above of=v7] (v6) {\(v_6\)};
            \node[random, draw, right of=v7] (v9) {\(v_9\)};
            \node[random, draw, above of=v9] (v8) {\(v_8\)};
    
            \node[state, right of=v9, xshift=6.5mm] (v11) {\(v_{11}\)};
            \node[square, draw, above of=v11] (v10) {\(v_{10}\)};
            
            \node[state, right of=v11, xshift=1mm] (v13) {\(v_{13}\)};
            \node[random, draw, above of=v13] (v12) {\(v_{12}\)};
            \node[square, draw, right of=v13] (v14) {\(v_{14}\)};
    
            \draw[rounded corners=4mm] (-1.2, -0.6) rectangle (0.6, 2.3);
            \draw[rounded corners=4mm] (0.9, -0.6) rectangle (4.6, 2.3);
            \draw[rounded corners=4mm] (4.9, -0.6) rectangle (8.5, 2.3);
            \draw[rounded corners=4mm] (8.8, -0.6) rectangle (10.1, 2.3);
            \draw[rounded corners=4mm] (10.4, -0.6) rectangle (13.7, 2.3);
            \draw 
                  (v1) edge[loop left] node[below, yshift=-1mm]{\small \(\EdgeValues{-2}{}\)} (v1)
            ;
            
            \draw
                  (v3) edge[left, pos=0.3] node{\small \(\EdgeValues{2}{}\)} (v2)
                  (v3) edge[below, pos=0.3] node{\small \(\EdgeValues{7}{}\)} (v5)
                  
                  (v4) edge[above left, pos=0.3] node[xshift=1mm, yshift=-0.5mm]{\(\EdgeValues{}{}\)} (v5)
            ;
            \draw[transform canvas={xshift=-0.4mm, yshift=+0.4mm}]
                  (v3) edge[above left, pos=0.2] node[xshift=1.2mm, yshift=-0.2mm]{\small \(\EdgeValues{2}{}\)} (v4)
            ;
            \draw[transform canvas={xshift=+0.4mm, yshift=-0.4mm}]
                  (v4) edge[below right, pos=0.2] node[xshift=-1.2mm, yshift=0.2mm]{\small \(\EdgeValues{0}{}\)} (v3)
            ;
    
            \draw
                  (v6) edge[loop left, pos=0.3] node[xshift=2mm, yshift=-2mm]{\small \(\EdgeValues{0}{}\)} (v6)
                  (v6) edge[above, pos=0.3] node[xshift=0.5mm, yshift=-0.5mm]{\small \(\EdgeValues{1}{}\)} (v8)
    
                  (v7) edge[left, pos=0.3] node[xshift=0.5mm, yshift=0mm]{\small \(\EdgeValues{-3}{}\)} (v6)
                  (v7) edge[below, pos=0.3] node[xshift=1mm, yshift=0mm]{\small \(\EdgeValues{-6}{}\)} (v9)
                 
                  (v9) edge[left, pos=0.3] node[xshift=1.4mm, yshift=-1.5mm]{\small \(\EdgeValues{0}{.9}\)} (v6)
                  (v9) edge[right, pos=0.3] node[xshift=-0.2mm, yshift=-1mm]{\small \(\EdgeValues{5}{.1}\)} (v8)
            ;
            
            \draw[transform canvas={xshift=+0.6mm}]
                  (v10) edge[right, pos=0.3] node{\(\EdgeValues{2}{}\)} (v11)
            ;
            \draw[transform canvas={xshift=-0.6mm}]
                  (v11) edge[left, pos=0.3] node{\(\EdgeValues{0}{}\)} (v10)
            ;
    
            \draw
                  (v12) edge[loop right, pos=0.3] node[xshift=1mm]{\small \(\EdgeValues{9}{.3}\)} (v12)
                  (v12) edge[auto, pos=0.3] node[xshift=-0.5mm]{\small \(\EdgeValues{-6}{.7}\)} (v13)
    
                  (v13) edge[below, pos=0.3] node{\small \(\EdgeValues{-4}{}\)} (v14)
    
                  (v14) edge[loop right] node[below, xshift=-0.2mm, yshift=-1.2mm]{\small \(\EdgeValues{2}{}\)}(v14)
            ;
            
            \draw
                  (v2) edge[loop left, pos=0.3] node[below, xshift=1mm]{\small \(\EdgeValues{0}{1}\)} (v2)
                  (v5) edge[loop right, pos=0.3] node[above, yshift=0mm]{\small \(\EdgeValues{0}{1}\)} (v5)
                  (v8) edge[loop right, pos=0.3] node[above, xshift=1mm]{\small \(\EdgeValues{0}{1}\)} (v8)
            ;
        \end{tikzpicture}
        }
        \caption{Restrictions  \(\Game_{\ValClass{1}}\),  \(\Game_{\ValClass{2}}\),  \(\Game_{\ValClass{3}}\), 
        \(\Game_{\ValClass{4}}\),
        and \(\Game_{\ValClass{5}}\)
        of the game shown in \Cref{fig:swmp-example} for the vector \(\ValVector = (-2, -1, -1, -1, -1, 0, 0, 0, 0, 1, 1, 2, 2, 2)\).}
        \label{fig:restriction-game-example}
    \end{figure}

    For the game in \Cref{fig:swmp-example},
    let \(\ValVector = (-2, -1, -1, -1, -1, 0, 0, 0, 0, 1, 1, 2, 2, 2)\) be a value vector for vertices \(v_1, v_2, \ldots, v_{14}\) respectively.
    Since \(\ValVector\) has five distinct values, we have \(\ValClassCount = 5\), and the five \(\ValVector\)-classes are \(\ValClass{1} = \{v_1\}\), \(\ValClass{2} = \{v_2, v_3, v_4, v_5\}\), \(\ValClass{3} = \{v_6, v_7, v_8, v_9\}\), \(\ValClass{4} = \{v_{10}, v_{11}\} \), and \(\ValClass{5} = \{v_{12}, v_{13}, v_{14}\}\) with values \(\ClassVal{1} = -2\), \(\ClassVal{2} = -1\), \(\ClassVal{3} = 0\), \(\ClassVal{4} = 1\), and \(\ClassVal{5} = 2\) respectively.
    Out of the five \kl{probabilistic vertices} \(v_2\), \(v_5\), \(v_8\), \(v_{9}\) and \(v_{12}\), we see that \(v_2\), \(v_5\), and \(v_{8}\) are \kl{boundary vertices} while \(v_9\) and \(v_{12}\) are not. 
    Thus, \(\Bnd{2} = \{v_2, v_5\}\), \(\Bnd{3} = \{v_8\}\), and \(\Bnd{1} = \Bnd{4} = \Bnd{5} = \emptyset\).
    We show the restrictions \(\Game_{\ValClass{i}}\) in \Cref{fig:restriction-game-example}.
    \lipicsEnd
\end{example}

Let \(\Objective\) be a bounded \kl{prefix-independent} \kl{quantitative objective}.
Analogous to the notation of a general value vector \(\ValVector\), we describe notations for the expected \(\Objective\)-value vector consisting of the expected \(\Objective\)-values of vertices in~\(\Vertices\).
For all vertices \(v \in \Vertices\), let \(\ValOpt{v}\) denote \(\ExpValue{v}{\Objective}\), the expected \(\Objective\)-value of vertex \(v\), and let \(\ValVectorOpt = (\ValOpt{v})_{v \in \Vertices}\)  denote the expected \(\Objective\)-value vector.
Let \(\OptClass{i}\) denote the \(i^{\text{th}}\) \(\ValVectorOpt\)-class and let \(\ClassValOpt{i}\) denote the \(\ValVectorOpt\)-value of \(\OptClass{i}\).

Given a vector \(\ValVector\), it follows from \Cref{prop:expected-wmp-value-out-neighbours} that the following is a necessary (but not sufficient) condition for \(\ValVector\) to be the expected \(\Objective\)-value vector \(\ValVectorOpt\).
\begin{itemize}
    \item \AP \intro *\ConditionBellman\ condition: 
    for every vertex \(v \in \Vertices\), the following Bellman equations hold
    \[
        \Val{v} = 
        \begin{cases}
            \max_{v' \in \OutNeighbours{v}} \Val{v'}
            & \text{if } v \in \VerticesMain, \\
            \min_{v' \in \OutNeighbours{v}} \Val{v'} 
            & \text{if } v \in \VerticesAdversary,  \\
            \sum_{v' \in \OutNeighbours{v}} \ProbabilityFunction(v)(v') \cdot \Val{v'} & \text{if } v \in \kl{\VerticesRandom}.
        \end{cases}
    \]
\end{itemize}
\subparagraph*{Consequences of the \ConditionBellman\ condition.}
We now see some properties of value vectors \(\ValVector\) that satisfy the \ConditionBellman\ condition.
Since \kl{boundary vertices} are \kl{probabilistic vertices}, the following is immediate.
\begin{proposition}%
\label{prop:boundary-vertices-lower-higher}
    Let \(\ValVector\) be a value vector satisfying the \ConditionBellman\ condition. 
    Then for all \(1 \le i \le \ValClassCount\), for all \(v \in \Bnd{i}\), there exists an out-neighbour of \(v\) with \(\ValVector\)-value less than \(\ClassVal{i}\) and there exists an out-neighbour of \(v\) with \(\ValVector\)-value greater than \(\ClassVal{i}\).
    Formally, there exist \(1 \le i_1, i_2 \le \ValClassCount\) such that \(\ClassVal{i_1} < \ClassVal{i} < \ClassVal{i_2}\) and \(\OutNeighbours{v} \intersection \ValClass{i_1} \ne \emptyset\) and \(\OutNeighbours{v} \intersection \ValClass{i_2} \ne \emptyset\). 
\end{proposition}
A corollary of \Cref{prop:boundary-vertices-lower-higher} is that the \(\ValVector\)-classes with the smallest and the biggest \(\ValVector\)-values have no \kl{boundary vertices}.
Note that there may also exist \(\ValVector\)-classes other than these that do not contain \kl{boundary vertices} (see \(\ValClass{4}\) in \Cref{ex:vc-notations}).
Next, we see that the \ConditionBellman\ condition entails that each restriction \(\Game_{\ValClass{i}}\) is a \kl{stochastic game}.
\begin{restatable}{proposition}{BellmanStochasticGame}
\label{prop:bellman-is-stochastic-game}
If \(\ValVector\) is a value vector that satisfies the \ConditionBellman\ condition, then for all \(1 \le i \le \ValClassCount\), we have that \(\Game_{\ValClass{i}}\) is a \kl{stochastic game}.
\end{restatable}
\begin{proof}
We need to show that each vertex in \(\Game_{\ValClass{i}}\) has an out-neighbour in \(\Game_{\ValClass{i}}\) and the probability distribution over the out-neighbours of probabilistic vertices in \(\Game_{\ValClass{i}}\) adds up to~\(1\).
From the \ConditionBellman\ condition, it follows that each non-probabilistic vertex \(v\) in \(\Game_{\ValClass{i}}\) has at least one out-neighbour \(v'\) with the same \(\ValVector\)-value as \(v\), and thus \(v'\) belongs to \(\Game_{\ValClass{i}}\). 
The \kl{boundary vertices} in \(\Game_{\ValClass{i}}\) have themself as an out-neighbour. 
Finally, the out-neighbours of non-boundary probabilistic vertices in \(\Game_{\ValClass{i}}\) are the same as in \(\Game\). 
Hence, \(\Game_{\ValClass{i}}\) is a \kl{stochastic game}.
\end{proof}

In \Cref{lem:eventually-no-boundary-vertices}, we make a  crucial observation about long-run behaviours of \kl{plays} in \(\Game\), 
which is that either player can ensure with probability~\(1\) that the token eventually reaches an \(\ValVector\)-class from which it does not exit.
This follows from the Borel-Cantelli lemma~\cite{Durrett2010} due to the fact that  there is a positive probability to reach an \(\ValVector\)-class without \kl{boundary vertices} following a finite number of edges out of \kl{boundary vertices}.
\begin{restatable}{proposition}{EventuallyNoBoundaryVertices}
\label{lem:eventually-no-boundary-vertices}
    Let \(\ValVector\) be a value vector satisfying the \ConditionBellman\ condition. 
    Suppose the strategy of \(\PlayerI\) (\(\i \in \{1, 2\}\)) is such that each time the token reaches a vertex \(v \in \Vertices[\i]\), (s)he moves the token to a vertex \(v'\) in the same \(\ValVector\)-class as \(v\).
    Then, with probability~1, the token eventually reaches a class \(\ValClass{j}\) for some \(1 \le j \le \ValClassCount\) from which it never exits. 
\end{restatable}
\begin{proof}
    We prove this for \(\i = \Main\). 
    The case of \(\i = \Adversary\) is analogous.
    If the token exits an \(\ValVector\)-class, it either exits from a \kl{boundary vertex} or from a \(\PlayerAdversary\) vertex.
    From \Cref{prop:boundary-vertices-lower-higher}, each time the token reaches a \kl{boundary vertex}, it has a positive probability of entering a class with a smaller \(\ValVector\)-value and a positive probability of entering a class with a greater \(\ValVector\)-value. 
    Furthermore, since \(\ValVector\) satisfies the \ConditionBellman\ condition, 
    \(\PlayerAdversary\) can only move the token to a class with a greater \(\ValVector\)-value. 
    Since there are \(\ValClassCount\) \(\ValVector\)-classes, out of which some do not have \kl{boundary vertices} (in particular the extremal \(\ValVector\)-classes do not have \kl{boundary vertices}), we get that starting from any vertex in the game, it is the case that with probability at least \((\ProbabilityFunction_{\min})^{\ValClassCount}\),  after changing classes at most \(\ValClassCount\) times, the token enters a class from which it never exits. 
    Here, \(\ProbabilityFunction_{\min}\) is the minimum probability over all edges in the game.
    By the second Borel-Cantelli lemma~\cite{Durrett2010}, in the infinite \kl{play}, with probability~\(1\), the token eventually enters a class from which it never exits. 
\end{proof}
Finally, we define the notion of \AP \intro{trap subgames} of \(\Game_{\ValClass{i}}\) which will be used in the subsequent discussion.
We denote by \(\AP \intro *\PositiveSubgame{\Main}{\ValClass{i}}\) the \(\PlayerMain\) positive attractor set of \(\Bnd{i}\), i.e., the set of vertices in \(\Game_{\ValClass{i}}\) that are positive winning for \(\PlayerMain\) for the \(\kl{\ReachObj}(\Bnd{i})\) objective.
The complement \(\AP \intro *\TrapSubgame{\Main}{\ValClass{i}} = \ValClass{i} \setminus \PositiveSubgame{\Main}{\ValClass{i}}\) is a \kl{trap} for \(\PlayerMain\) in \(\Game_{\ValClass{i}}\), 
and with abuse of notation, we use the same symbol \(\TrapSubgame{\Main}{\ValClass{i}}\) to denote the subset of \(\Game_{\ValClass{i}}\) as well as the \kl{trap subgame}.
We note that if \(\ValClass{i}\) does not have \kl{boundary vertices}, that is, if \(\Bnd{i} = \emptyset\), then it holds that  \(\PositiveSubgame{\Main}{\ValClass{i}}  = \emptyset\) and \(\TrapSubgame{\Main}{\ValClass{i}} = \ValClass{i}\).
We can analogously define \(\PositiveSubgame{\Adversary}{\ValClass{i}}\) and \(\TrapSubgame{\Adversary}{\ValClass{i}}\) for \(\PlayerAdversary\).
Given \(\Game_{\ValClass{i}}\), these sets can be computed in polynomial time using attractor computations~\cite{FBB24}. 

\begin{example}%
\label{ex:apt}
    We compute these sets for the restrictions shown in \Cref{fig:restriction-game-example}.
    For \(i \in \{1, 4, 5\}\), since \(\Bnd{i}\) is empty, we have that \(\TrapSubgame{\Main}{\ValClass{i}} = \ValClass{i}\) and \(\PositiveSubgame{\Main}{\ValClass{i}} = \emptyset\).
    For \(\ValClass{2}\), we have that \(\PositiveSubgame{\Main}{\ValClass{2}} = \ValClass{2}\), and \(\TrapSubgame{\Main}{\ValClass{2}} = \emptyset\).
    For \(\ValClass{3}\), we have that \(\TrapSubgame{\Main}{\ValClass{3}} = \{v_6\}\), \(\PositiveSubgame{\Main}{\ValClass{3}} = \{v_7, v_8, v_9\}\).
    Analogously, for \(\PlayerAdversary\), we have the following:
    For \(i \in \{1, 4, 5\}\), since \(\Bnd{i}\) is empty, we also have that \(\TrapSubgame{\Adversary}{\ValClass{i}} = \ValClass{i}\) and \(\PositiveSubgame{\Adversary}{\ValClass{i}} = \emptyset\).
    For \(i \in \{2, 3\}\), we have that \(\PositiveSubgame{\Adversary}{\ValClass{i}} = \ValClass{i}\), and thus, \(\TrapSubgame{\Adversary}{\ValClass{i}} = \emptyset\).
    \lipicsEnd
\end{example}

\subsection{Characterization of the value vector}
We describe in \Cref{thm:six-conditions} a necessary and sufficient set of conditions for a given vector \(\ValVector\) to be equal to the expected \(\Objective\)-value vector \(\ValVectorOpt\).
In addition to \ConditionBellman, \Cref{thm:six-conditions} makes use of two more conditions, which we define before stating the theorem.
\begin{itemize}
    \item \AP \intro *\ConditionLB\ condition: 
    for all \(1 \le i \le \ValClassCount\), \(\PlayerMain\) wins \(\ThresholdObjective{> \ClassVal{i} - \GranularityBound{\Objective}}\) almost surely in the \kl{trap subgame} \(\TrapSubgame{\Main}{\ValClass{i}}\) from all vertices in \(\TrapSubgame{\Main}{\ValClass{i}}\).
    \item \AP \intro *\ConditionUB\ condition: 
    for all \(1 \le i \le \ValClassCount\), \(\PlayerAdversary\) wins \(\ThresholdObjective{< \ClassVal{i} + \GranularityBound{\Objective}}\) almost surely in the \kl{trap subgame} \(\TrapSubgame{\Adversary}{\ValClass{i}}\) from all vertices in \(\TrapSubgame{\Adversary}{\ValClass{i}}\).
\end{itemize}

\begin{theorem}%
\label{thm:six-conditions}
    The only vector \(\ValVector\), whose every component has denominator at most \(\DenBound{\Objective}\), that satisfies \ConditionBellman, \ConditionLB, and \ConditionUB\ is the expected \(\Objective\)-value vector~\(\ValVectorOpt\).
\end{theorem}
\begin{proof}
    We show in \Cref{lem:e-satisfies-six-conditions} that \(\ValVectorOpt\) satisfies the three conditions. 
    We show in \Cref{lem:fwmpl-almost-sure-to-expected} that if \(\ValVector\) is a vector that satisfies the three conditions, then \(\ValVector\) is less than \(\GranularityBound{\Objective}\) distance away from \(\ValVectorOpt\), that is, \(\ValVectorOpt - \GranularityBound{\Objective}  < \ValVector < \ValVectorOpt + \GranularityBound{\Objective} \).
    In particular, if each component of \(\ValVector\) can be written as \(\frac{p}{q}\), where \(p\), \(q\) are both integers and \(q\) is at most \(\DenBound{\Objective}\), then it follows that \(\ValVector\) is equal to \(\ValVectorOpt\).
\end{proof}

In the rest of the section, we prove Lemmas~\ref{lem:e-satisfies-six-conditions} and~\ref{lem:fwmpl-almost-sure-to-expected} used in the proof of \Cref{thm:six-conditions}. 
\begin{lemma}%
\label{lem:e-satisfies-six-conditions}
    The expected \(\Objective\)-value vector \(\ValVectorOpt\) satisfies the three conditions in \Cref{thm:six-conditions}.
\end{lemma}
\begin{proof}
    The fact that \(\ValVectorOpt\) satisfies \(\ConditionBellman\) follows  directly from \Cref{prop:expected-wmp-value-out-neighbours}.
    We show that \(\ConditionLB\) holds for \(\ValVectorOpt\).
    The proof for \(\ConditionUB\) is analogous.
        
    Suppose for the sake of contradiction that \(\ConditionLB\) does not hold, 
    that is, there exists \(1 \le i \le \ValClassCount[\ValVectorOpt]\) and a vertex \(v\) in 
    \(\TrapSubgame{\Main}{\OptClass{i}}\)
    such that \(\PlayerAdversary\) has a positive winning strategy from \(v\) for the 
    \(\ThresholdObjective{\le \ClassValOpt{i} - \GranularityBound{\Objective}}\) 
    objective in \(\TrapSubgame{\Main}{\OptClass{i}}\).
    Since \(\ThresholdObjective{\le \ClassValOpt{i}- \GranularityBound{\Objective}}\) 
    is a \kl{prefix-independent} objective, from~\cite[Theorem 1]{Cha07} we have that there exists another vertex 
    \(v'\) in \(\TrapSubgame{\Main}{\OptClass{i}}\)
    such that \(\PlayerAdversary\) has an almost-sure winning strategy from \(v'\) for the same objective
    \(\ThresholdObjective{\le \ClassValOpt{i} - \GranularityBound{\Objective}}\) 
    in \(\TrapSubgame{\Main}{\OptClass{i}}\).
    If \(\PlayerAdversary\) follows this strategy from \(v'\) in the \emph{original game} \(\Game\),
    then one of the following two cases holds
    \begin{itemize}
        \item \(\PlayerMain\) always moves the token to a vertex in \(\OptClass{i}\). 
        Since \(v'\) is in the trap \(\TrapSubgame{\Main}{\OptClass{i}}\) for \(\PlayerMain\) in \(\Game_{\OptClass{i}}\),
        \(\PlayerAdversary\) can force the token to remain in
        \(\TrapSubgame{\Main}{\OptClass{i}}\) forever, and follow the almost-sure winning strategy to ensure that with probability~1, the outcome satisfies the objective \(\ThresholdObjective{\le \ClassValOpt{i} - \GranularityBound{\Objective} }\).
        \item \(\PlayerMain\) eventually moves the token to a vertex out of \(\OptClass{i}\).
        Since \(\ValVectorOpt\) satisfies \ConditionBellman, the token moves to an \(\ValVectorOpt\)-class with a smaller value than \(\ClassValOpt{i}\).
    \end{itemize}
    In both cases, the expected \(\Objective\)-value of the outcome is less than \(\ClassValOpt{i}\).
    This is a contradiction since \(v' \in \OptClass{i}\), and 
    the expected \(\Objective\)-value of every vertex in \(\OptClass{i}\) is equal to \(\ClassValOpt{i}\).
\end{proof}

\begin{restatable}{lemma}{AlmostSureToExpected}
\label{lem:fwmpl-almost-sure-to-expected}
    If a vector \(\ValVector\) satisfies the three conditions in \Cref{thm:six-conditions}, then \(\ValVectorOpt - \GranularityBound{\Objective}  < \ValVector < \ValVectorOpt +  \GranularityBound{\Objective} \). 
    In particular, we have the following:
    \begin{itemize}
        \item If \(\ValVector\) satisfies the \ConditionBellman\ and \ConditionLB\ conditions, then \(\ValVectorOpt > \ValVector - \GranularityBound{\Objective} \). 
        \item If \(\ValVector\) satisfies the \ConditionBellman\ and \ConditionUB\ conditions, then \(\ValVectorOpt < \ValVector + \GranularityBound{\Objective} \). 
    \end{itemize}
\end{restatable}
\begin{proof}
We prove the first case.
The proof for the second case follows by symmetry, that is, we 
essentially 
replace \(\PlayerMain\) by \(\PlayerAdversary\), and \(\ThresholdObjective{> \ClassVal{i} - \GranularityBound{\Objective} }\) by \(\ThresholdObjective{< \ClassVal{i} + \GranularityBound{\Objective} }\). 
First, we describe an optimal strategy \(\Strategy[\Main]^{*}\) of \(\PlayerMain\) and give a sketch of its optimality.
We show the full proof later.

Since \(\ValVector\) satisfies the \ConditionLB\ condition,  we have that for all \(1 \le i \le \ValClassCount\), \(\PlayerMain\) has an  an almost-sure winning strategy \(\TrapStrategy{\ValClass{i}}\) in the \kl{trap subgame} \(\TrapSubgame{\Main}{\ValClass{i}}\) to win the objective   \(\ThresholdObjective{> \ClassVal{i} - \GranularityBound{\Objective}}\)  in \(\TrapSubgame{\Main}{\ValClass{i}}\) almost surely from all vertices in \(\TrapSubgame{\Main}{\ValClass{i}}\).
From the definition of \(\PositiveSubgame{\Main}{\ValClass{i}}\), 
\(\PlayerMain\) has a positive winning strategy \(\ReachStrategy{\ValClass{i}}\) 
in the restricted game \(\Game_{\ValClass{i}}\)
from vertices in \(\PositiveSubgame{\Main}{\ValClass{i}}\) for the \(\kl{\ReachObj}(\Bnd{i})\) objective.
By following \(\ReachStrategy{\ValClass{i}}\), the token either reaches \(\Bnd{i}\) with positive probability, or ends up in \(\TrapSubgame{\Main}{\ValClass{i}}\) from where \(\PlayerAdversary\) can ensure that the token never leaves \(\TrapSubgame{\Main}{\ValClass{i}}\).
Using these strategies of \(\PlayerMain\) in \(\Game_{\ValClass{i}}\), we construct a strategy \(\Strategy[\Main]^{*}\)
of \(\PlayerMain\) that is optimal for expected \(\Objective\)-value in the \emph{original game} \(\Game\):
As long as the token is in the class \(\ValClass{i}\) in \(\Game\), the strategy \(\Strategy[\Main]^{*}\) mimics \(\TrapStrategy{\ValClass{i}}\) if the token is in \(\TrapSubgame{\Main}{\ValClass{i}}\) and  \(\Strategy[\Main]^{*}\) mimics \(\ReachStrategy{\ValClass{i}}\) if the token is in \(\PositiveSubgame{\Main}{\ValClass{i}}\).

Note that whenever the token is on a vertex \(v \in \VerticesMain\) in \(\ValClass{i}\), the strategy \(\Strategy[\Main]^{*}\) always moves the token to a vertex \(v'\) in same \(\ValVector\)-class \(\ValClass{i}\) as \(v\) (i.e. a token only exits an \(\ValVector\)-class from a \(\PlayerAdversary\) vertex or from a \kl{boundary vertex}), and thus, \Cref{lem:eventually-no-boundary-vertices} holds.
Whenever the token exits a class \(\ValClass{i}\) to reach a different class \(\ValClass{i'}\), then as long as the token remains in \(\ValClass{i'}\), the strategy \(\Strategy[\Main]^{*}\) follows \(\TrapStrategy{\ValClass{i'}}\) if the token is in \(\TrapSubgame{\Main}{\ValClass{i'}}\), and 
    \(\Strategy[\Main]^{*}\) follows \(\ReachStrategy{\ValClass{i'}}\)
    if the token is in \(\PositiveSubgame{\Main}{\ValClass{i'}}\).

Since \Cref{lem:eventually-no-boundary-vertices} holds, we have that for all strategies of \(\PlayerAdversary\), with probability~\(1\), the token eventually reaches an \(\ValVector\)-class \(\ValClass{j}\) from which it never exits.
Moreover, the strategy \(\Strategy[\Main]^{*}\) ensures that with probability~\(1\), the token eventually reaches \(\TrapSubgame{\Main}{\ValClass{j}}\) in \(\ValClass{j}\) from which it never leaves.
Because if not, then the token would visit vertices in \(\PositiveSubgame{\Main}{\ValClass{j}}\) infinitely often, having a fixed positive probability of reaching \(\Bnd{j}\) in every step because of \(\ReachStrategy{\ValClass{j}}\). 
Thus, with probability~\(1\), the token  would eventually reach \(\Bnd{j}\) from where it could escape to a different \(\ValVector\)-class, which contradicts the fact that the token stays in \(\ValClass{j}\) forever.

Since \(\Objective\) is \kl{prefix-independent}, the \(\Objective\)-value of a \kl{play} only depends on the trap \(\TrapSubgame{\Main}{\ValClass{j}}\) it ends up in. 
If the game begins from a vertex \(v \in \ValClass{i}\), 
then for \(1 \le j \le \ValClassCount\), let \(p_j\) denote the probability that the token ends up in the \kl{trap subgame} \(\TrapSubgame{\Main}{\ValClass{j}}\) from which it never exits.
Since \(\ValVector\) satisfies \ConditionBellman, we have that \(\sum_{j} p_j \ClassVal{j} = \ClassVal{i}\). 
Since \(\ValVector\) satisfies \ConditionLB, \(\PlayerMain\) has an almost-sure winning strategy for \(\ThresholdObjective{> \ClassVal{j} - \GranularityBound{\Objective}}\) in \(\TrapSubgame{\Main}{\ValClass{j}}\).
Thus, for all strategies \(\Strategy[\Adversary]\) of \(\PlayerAdversary\), the expected value of an outcome of \((\Strategy[\Main]^{*}, \Strategy[\Adversary])\) from \(v \in \ValClass{i}\) is greater than \(\sum_{j} p_j (\ClassVal{j} - \GranularityBound{\Objective})\), which is \(\ClassVal{i} - \GranularityBound{\Objective}\).
This holds for all vertices \(v\) in \(\Game\), giving us the desired result  \(\ValVectorOpt > \ValVector - \GranularityBound{\Objective} \).

\subparagraph{Full proof of optimality of \(\Strategy[\Main]^{*}\).}
We formally show that the strategy \(\Strategy[\Main]^{*}\) ensures that for all vertices \(v\) in the game \(\Game\), we have that \(\ValOpt{v}\), the expected \(\Objective\)-value of the outcome in \(\Game\) starting from \(v\), is greater than \(\Val{v} - \GranularityBound{\Objective}\).
Let \(\Strategy[\Adversary]\) be an arbitrary  strategy\footnote{
    In this proof, we proceed by assuming that \(\Strategy[\Adversary]\) is a deterministic strategy.
    If \(\Strategy[\Adversary]\) is randomized, then branching will occur not only at vertices in \(\VerticesRandom\) but also at vertices in \(\VerticesAdversary\) in the tree \(\TreeMC\).
    The remaining observations in the proof continue to hold. 
}
of \(\PlayerAdversary\).
Fixing the strategy profile \(\StrategyProfile = (\Strategy[\Main]^{*}, \Strategy[\Adversary])\) yields a (possibly infinite) Markov chain \(\Game_{\StrategyProfile}\), i.e., a probability distribution over the set of \kl{plays} in \(\Game\) starting at \(v\) consistent with the strategy profile \(\StrategyProfile\). 
We unfold \(\Game_{\StrategyProfile}\) to obtain an infinite rooted tree \(\TreeMC\) with the root being the initial vertex \(v\) in \(\Game\).
Each vertex in the tree \(\TreeMC\) is uniquely identified by the path from the root to that vertex, which corresponds to a unique prefix in the game \(\Game\).
For ease of presentation, we label each vertex of \(\TreeMC\) by only the last vertex of the prefix.
There is a one-to-one correspondence between the set of infinite paths in \(\TreeMC\) starting at the root \(v\) and the set of \kl{plays} in \(\Game\) starting from \(v\) consistent with \(\StrategyProfile\).
For each vertex \(u\) in \(\TreeMC\), if \(u \in \VerticesMain \union \VerticesAdversary\), then \(u\) has exactly one child in \(\TreeMC\) determined by the strategy profile \(\StrategyProfile\), and if \(u \in \VerticesRandom\), then the probability distribution over the out-neighbours of \(u\) in \(\TreeMC\) is the same as in \(\Game_{\StrategyProfile}\). 
Thus, the only branching that occurs in \(\TreeMC\) is at probabilistic vertices.

We say a vertex \(u\) in the tree \(\TreeMC\) is \emph{final} if all descendants of \(u\)  in \(\TreeMC\) (i.e., all vertices reachable from \(u\) in \(\TreeMC\)) belong to the same \(\ValVector\)-class as \(u\). 
Once the token reaches a final vertex \(u\) in \(\TreeMC\), i.e., once the token visits the prefix corresponding to \(u\) in \(\Game\), then the token never exits the \(\ValVector\)-class that \(u\) belongs to.
That is, \(\PlayerAdversary\) never moves the token to a different \(\ValVector\)-class and the token never reaches a \kl{boundary vertex} either. 
Note that, in particular, every vertex in the \(\ValVector\)-class with the greatest value is a final vertex since the class does not have \kl{boundary vertices}, and all out-neighbours of \(\PlayerAdversary\) vertices in the class belong to the same class.

We trim final vertices in \(\TreeMC\), i.e., for each infinite path in \(\TreeMC\) beginning at the root, we keep the first final vertex \(u\) appearing in the path and delete all descendants of \(u\). 
Let the trimmed tree be denoted by \(\TreeTrimmed\). 
For \(d \ge 1\), let \(\TreeTruncated\) denote the tree \(\TreeTrimmed\) truncated to depth \(d\). 
That is, for all paths in \(\TreeTrimmed\) of length greater than \(d\), delete all vertices that are at a distance greater than \(d\) from the root. 
Since \(\TreeTruncated\) is finite, starting from the root \(v\) of \(\TreeTruncated\), with probability~\(1\), one of the leaves of \(\TreeTruncated\) is reached.
\begin{equation}\label{eq:reach-leaves-probability-one}
    \sum_{u \text{ is leaf in } \TreeTruncated} p_{v \to u} = 1
\end{equation}
Here, \(p_{v \to u}\) denotes the probability of reaching \(u\) from the root \(v\) in \(\TreeTruncated\), that is, the product of the probabilities along the unique path from \(v\) to \(u\).

\subparagraph{\(\ValVector\)-value of the root in terms of leaves.}
Let \(u\) be a non-final vertex in \(\TreeTrimmed\). 
We have the following relation between the \(\ValVector\)-values of \(u\) and its children in \(\TreeTrimmed\). 
\begin{itemize}
    \item If \(u \in \VerticesMain\), then \(\Val{u}= \Val{u'}\), where \(u'\) is the child of \(u\) in \(\TreeTrimmed\).
    This holds since the strategy \(\Strategy[\Main]^{*}\) of \(\PlayerMain\) always returns a successor vertex in the same \(\ValVector\)-class.
    \item If \(u \in \VerticesAdversary\), then \(\Val{u} \le \Val{u'}\), where \(u'\) is the child of \(u\) in \(\TreeTrimmed\).
    This follows since \(\ValVector\) satisfies the \ConditionBellman\ condition of \Cref{thm:six-conditions}.
    \item If \(u \in \VerticesRandom\), then \(\Val{u} = \sum_{u' \in \OutNeighbours{u}} \ProbabilityFunction(u)(u') \cdot \Val{u'}\). 
    If \(u\) is a \kl{boundary vertex}, then the equation holds since \(\ValVector\) satisfies the \ConditionBellman\ condition of \Cref{thm:six-conditions}.
    Otherwise, \(u\) is a non-boundary vertex and all out-neighbours of \(u\) are in the same \(\ValVector\)-class as \(u\), and thus we have \(\Val{u} = \Val{u'}\) for all out-neighbours \(u'\) of \(u\). 
\end{itemize}
By induction on the length of paths from the root \(v\), we get the following relation in \(\TreeTruncated\) for all \(d \ge 1\). 
\begin{equation}\label{eq:guessed-value-root-leaves}
\Val{v} \le \sum_{u \text{ is leaf in } \TreeTruncated} p_{v \to u} \cdot \Val{u}
\end{equation}

\subparagraph{\(\ValVectorOpt\)-value of the root in terms of leaves.}
For a vertex \(u\) in the tree \(\TreeTrimmed\), 
let \(\ValOutcomeOpt{u}{\StrategyProfile}\) denote the expected \(\Objective\)-value of the outcome of \(\StrategyProfile\) following the prefix that is the unique path from the root \(v\) to vertex \(u\) in \(\TreeTrimmed\).
The \(\ValVectorOpt\)-value of a non-final vertex \(u\) in the tree is equal to the weighted average of the \(\ValVectorOpt\)-values of its children in the tree.
Formally, we have the following from the Bellman equations in \Cref{prop:expected-wmp-value-out-neighbours}.
\begin{itemize}
    \item If \(u \in \VerticesMain \union \VerticesAdversary\), then \(\ValOutcomeOpt{u}{\StrategyProfile} = \ValOutcomeOpt{u'}{\StrategyProfile}\), where \(u'\) is the child of \(u\) in \(\TreeTrimmed\).
    \item If \(u \in \VerticesRandom\), then \(\ValOutcomeOpt{u}{\StrategyProfile} = \sum_{u' \in \OutNeighbours{u}} \ProbabilityFunction(u)(u') \cdot \ValOutcomeOpt{u'}{\StrategyProfile}\). 
\end{itemize}
By induction on the length of paths from the root \(v\), we get the following relation in \(\TreeTruncated\) for all \(d \ge 1\). 
\begin{equation}\label{eq:expected-value-root-leaves}
\ValOutcomeOpt{v}{\StrategyProfile} = \sum_{u \text{ is leaf in } \TreeTruncated} p_{v \to u} \cdot \ValOutcomeOpt{u}{\StrategyProfile}
\end{equation}

\subparagraph{\(\ValVectorOpt\)-values of final vertices.}
Recall that once the token reaches a final vertex \(u\), the strategy \(\Strategy[\Main]^{*}\) plays the almost-sure winning strategy for the \kl{threshold objective} \(\ThresholdObjective{> \Val{u} - \GranularityBound{\Objective} }\). 
Thus,
\begin{equation}\label{eq:final-vertex-value}
\text{if } u \text{ is a final vertex, then } \ValOutcomeOpt{u}{\Strategy} > \Val{u} - \GranularityBound{\Objective} .
\end{equation}

\subparagraph{Eliminating non-final leaves.}
If \(\TreeTrimmed\) is finite, then there exists \(d \ge 1\) such that we have \(\TreeTrimmed = \TreeTruncated\).
Every leaf of \(\TreeTruncated\) is a final vertex. 
Thus, from  \eqref{eq:guessed-value-root-leaves}, \eqref{eq:expected-value-root-leaves}, and \eqref{eq:final-vertex-value}, we get that \(\ValOutcomeOpt{v}{\StrategyProfile} > \Val{v} - \GranularityBound{\Objective}\) and we are done. 
Otherwise, suppose \(\TreeTrimmed\) is not finite. 
Then, for all \(d \ge 1\), we have that some of the leaves of \(\TreeTruncated\) are not final. 
Let \(\PayoffFunction_{\min}\) denote the minimum \(\Objective\)-value of a \kl{play} occurring in \(\Game\). 
This exists because \(\Objective\) is bounded.
\begin{equation}\label{eq:non-final-vertex-value}
\ValOutcomeOpt{u}{\StrategyProfile} \ge \PayoffFunction_{\min}
\end{equation}
We split the sum in~\eqref{eq:expected-value-root-leaves} depending on whether \(u\) is final in \(\TreeTruncated\) or not.
\begin{align}
    \ValOutcomeOpt{v}{\StrategyProfile} 
    &= \sum_{\substack{u \text{ is final}\\ \text{ in } \TreeTruncated}}  p_{v \to u} \cdot \ValOutcomeOpt{u}{\StrategyProfile} +  \sum_{\substack{u \text{ is non-final}\\ \text{leaf in } \TreeTruncated}} p_{v \to u} \cdot \ValOutcomeOpt{u}{\StrategyProfile} \\
    &> \sum_{\substack{u \text{ is final}\\ \text{ in } \TreeTruncated}} p_{v \to u} \cdot (\Val{u} - \GranularityBound{\Objective} )+  \sum_{\substack{u \text{ is non-final}\\ \text{leaf in } \TreeTruncated}} p_{v \to u} \cdot \PayoffFunction_{\min}\label{eq:expected-value-lower-bound}
\end{align}
The inequality follows from~\eqref{eq:final-vertex-value} and~\eqref{eq:non-final-vertex-value}. 
Similarly, we can split the sum in~\eqref{eq:guessed-value-root-leaves} depending on whether \(u\) is final in \(\TreeTruncated\) or not. 
\begin{align}
    \Val{v} 
    &\le \sum_{\substack{u \text{ is final}\\ \text{ in } \TreeTruncated}}  p_{v \to u} \cdot \Val{u} + \sum_{\substack{u \text{ is non-final}\\ \text{leaf in } \TreeTruncated}} p_{v \to u} \cdot \Val{u}\label{eq:guessed-value-lower-bound}
\end{align}
Since the strategy \(\Strategy[\Main]^{*}\) never moves the token to a different \(\ValVector\)-class, from \Cref{lem:eventually-no-boundary-vertices}, we have that with probability~\(1\), the token eventually reaches a final vertex. 
From~\eqref{eq:reach-leaves-probability-one}, in \(\TreeTruncated\), the probability of reaching a leaf from the root \(v\) is~1. 
Thus, as \(d\) increases, the probability measure of final leaves in \(\TreeTruncated\) increases to \(1\) and the probability of non-final leaves in \(\TreeTruncated\) decreases to \(0\).
\begin{equation}\label{eq:limit-final-probability-one}
    \lim_{d \to \infty} \sum_{\substack{u \text{ is final}\\ \text{ in } \TreeTruncated}} p_{v \to u} = 1, \qquad
    \lim_{d \to \infty}
    \sum_{\substack{u \text{ is non-final}\\ \text{leaf in } \TreeTruncated}}
    p_{v \to u} = 0
\end{equation}
The following limits follow from~\eqref{eq:guessed-value-lower-bound} and~\eqref{eq:limit-final-probability-one}. 
\begin{equation}\label{eq:final-vertices-limit}
    \lim_{d \to \infty} \sum_{\substack{u \text{ is final}\\ \text{ in } \TreeTruncated}} p_{v \to u} \cdot \Val{u}  \ge \Val{v}, \quad
    \lim_{d \to \infty} \sum_{\substack{u \text{ is non-final}\\ \text{leaf in } \TreeTruncated}}p_{v \to u} \cdot \PayoffFunction_{\min} = 0
\end{equation}
Thus, from~\eqref{eq:final-vertices-limit}, we see that as \(d\) increases, the expression in~\eqref{eq:expected-value-lower-bound} becomes a better approximation of \(\ValOutcomeOpt{v}{\StrategyProfile}\).
As \(d \to \infty\), we have that \(\TreeTruncated \to \TreeTrimmed\), and \(\ValOutcomeOpt{v}{\StrategyProfile} > \Val{v} - \GranularityBound{\Objective} \).
Since this holds for arbitrary strategies of \(\PlayerAdversary\), we get that \(\ValOpt{v}\), the expected \(\Objective\)-value of the initial vertex \(v\), is at least \(\Val{v}\). 
\end{proof}

We also note that the optimal strategy \(\Strategy[\Main]^{*}\) always either follows an almost-sure winning strategy \(\TrapStrategy{\ValClass{i}}\) for the \kl{threshold objective} \(\ThresholdObjective{> \ClassVal{i} - \GranularityBound{\Objective}}\) or a positive winning strategy for a \(\kl{\ReachObj}\) objective.
Since there exist memoryless positive winning strategies for the \(\kl{\ReachObj}\) objective~\cite{Con92}, we have the following bound on the memory requirement of \(\Strategy[\Main]^{*}\).
\begin{corollary}%
\label{cor:memory-bound}
    The memory requirement of \(\Strategy[\Main]^{*}\) is at most the maximum over all \(1 \le i \le \ValClassCount\)  of the memory requirement of an almost-sure winning strategy \(\TrapStrategy{\ValClass{i}}\) for the \kl{threshold objective} \(\ThresholdObjective{> \ClassVal{i} - \GranularityBound{\Objective}}\).
    Moreover, if \(\TrapStrategy{\ValClass{i}}\) is a deterministic strategy, then so is \(\Strategy[\Main]^{*}\).
\end{corollary}

\subsection{Bounding the denominators in the value vector}
In this section, we discuss the problem of obtaining an upper bound \(\DenBound{\Objective}\) for the denominators of the expected \(\Objective\)-values of vertices \(\ClassValOpt{i}\) for a bounded \kl{prefix-independent} objective \(\Objective\) in a game \(\Game\).
In~\cite{CHH09}, the technique of value class is used to compute the values of vertices for Boolean \kl{prefix-independent} objectives.
It is stated without proof that the probability of satisfaction of a parity or a Streett objective~\cite{AG11} from each vertex can be written as \(\frac{p}{q}\) where \(q \le (\ProbabilityDenominator)^{4 \cdot \abs{\Edges}}\) and \(\ProbabilityDenominator\) is the maximum denominator over all edge probabilities in the game.
As such, we were not able to directly generalize this bound for the expectation of quantitative \kl{prefix-independent} objectives.
Instead, we make the following observations:
\begin{itemize}
    \item 
    Let \(\OptClass{i}\) be an \(\ValVectorOpt\)-class without \kl{boundary vertices}.
    If the token is in \(\OptClass{i}\) at some point in the \kl{play}, then since \(\ValVectorOpt\) satisfies the \ConditionBellman\ condition,  neither player has an incentive to move the token out of \(\OptClass{i}\). Since there are no \kl{boundary vertices} in \(\OptClass{i}\), the token does not exit \(\OptClass{i}\) from a probabilistic vertex either, and remains in \(\OptClass{i}\) forever.
    Thus, the value \(\ClassValOpt{i}\) of \(\OptClass{i}\) depends only on the internal structure of \(\OptClass{i}\).
    We denote by \(\DenBoundNoBoundaryVertex{\Objective}\) an upper bound on the denominators of values of \(\ValVectorOpt\)-classes without \kl{boundary vertices}.
    It is a simpler problem to find \(\DenBoundNoBoundaryVertex{\Objective}\) than to find \(\DenBound{\Objective}\),
    as each class without \kl{boundary vertices} can be treated as a subgame in which each vertex has the same expected \(\Objective\)-value, or equivalently, the subgame consists of exactly one \(\ValVectorOpt\)-class.
    \item 
    On the other hand, suppose \(\OptClass{i}\) is an \(\ValVectorOpt\)-class containing at least one \kl{boundary vertex}, and let \(v\) be a \kl{boundary vertex} in \(\OptClass{i}\).
    Then, since \(\ValVectorOpt\) satisfies the \ConditionBellman\ condition,
    we have \(\ValOpt{v} = \sum_{v' \in \OutNeighbours{v}} \ProbabilityFunction(v)(v') \cdot \ValOpt{v'}\), which is also the value \(\ClassValOpt{i}\) of \(\OptClass{i}\).
    Thus, \(\ClassValOpt{i}\) can be written in terms of the values of classes reachable from \(v\) in one step and the probabilities with which those classes are reached.
    In fact, 
    we construct in the proof of \Cref{thm:denominator-bound} a system of linear equations to show that the value of each \(\ValVectorOpt\)-class with \kl{boundary vertices} can be expressed solely in terms of transition probabilities of the outgoing edges from \kl{boundary vertices} and values of \(\ValVectorOpt\)-classes without \kl{boundary vertices}. 
\end{itemize}
The method to calculate \(\DenBoundNoBoundaryVertex{\Objective}\) depends on the specific objective; 
we illustrate as an example in \Cref{sec:window-mean-payoff-objectives} a way to obtain \(\DenBoundNoBoundaryVertex{\Objective}\) for a particular kind of objective called the window mean-payoff objective.
Once we know \(\DenBoundNoBoundaryVertex{\Objective}\) for an objective \(\Objective\), we can use \Cref{thm:denominator-bound} to obtain \(\DenBound{\Objective}\) in terms of \(\DenBoundNoBoundaryVertex{\Objective}\).
\begin{theorem}%
\label{thm:denominator-bound}
    The denominator of the value of each \(\ValVectorOpt\)-class in \(\Game\) is at most \( \DenBound{\Objective} = 2^{\abs{\Vertices}} \cdot \ProbabilityDenominator^{\abs{V}^3}  \cdot (\DenBoundNoBoundaryVertex{\Objective})^{\abs{\Vertices}}\). 
\end{theorem}
We note that this theorem implies that the number of bits required to write \(\DenBound{\Objective}\) is polynomial in the number of vertices in the game and in the number of bits required to write \(\DenBoundNoBoundaryVertex{\Objective}\).
We devote the rest of this section to the proof of \Cref{thm:denominator-bound}.
For ease of notation, we denote the number of \(\ValVectorOpt\)-classes in the game by \(k\) instead of \(\ValClassCount[\ValVectorOpt]\) for the rest of this section.
If every \(\ValVectorOpt\)-class has no \kl{boundary vertices}, then we have \(\DenBound{\Objective}\) equal to \(\DenBoundNoBoundaryVertex{\Objective}\) and we are done. 
So we assume there exists at least one class that contains \kl{boundary vertices}.
Let \(m \ge 1\) denote the number of \(\ValVectorOpt\)-classes with \kl{boundary vertices}, and therefore, there are \(k - m\) \(\ValVectorOpt\)-classes without \kl{boundary vertices}.
Since there always exists at least one \(\ValVectorOpt\)-class without \kl{boundary vertices} (\Cref{prop:boundary-vertices-lower-higher}), we have that \(m < k\).
Let \(B = \{1, 2, \ldots, m\}\) and \(C = \{m+1, \ldots, k\}\).
We index the \(\ValVectorOpt\)-classes such that each class with \kl{boundary vertices} has its index in \(B\)
and each class without \kl{boundary vertices} has its index in \(C\).
Furthermore, in the sets \(B\) and \(C\), the classes are indexed in increasing order of their values.
That is, for \(i, j\) both in \(B\) or both in \(C\), we have \(i < j\) if and only if \(\ClassValOpt{i} < \ClassValOpt{j}\).
We show bounds on the denominators of \(\ValVectorOpt\)-values of classes with \kl{boundary vertices}, i.e., \(\ClassValOpt{1}, \ldots, \ClassValOpt{m}\) in terms of  \(\ValVectorOpt\)-values of classes without \kl{boundary vertices}, i.e., \(\ClassValOpt{m+1}, \ldots, \ClassValOpt{k}\).

For all \(i \in B = \{1, 2, \ldots, m\}\), pick an arbitrary \kl{boundary vertex} from \(\Bnd[\OptSymbol]{i}\) and call this the \emph{representative vertex} \(u_i\) of \(\Bnd[\OptSymbol]{i}\). 
For all \(i \in B\) and \(j \in \{1, 2, \ldots, k\}\), let \(p_{i,j}\) denote the probability of reaching the class \(\OptClass{j}\) from \(u_i\) in one step.
Since~\(\ValVectorOpt\) satisfies the \ConditionBellman\ condition, we have that \(\sum_{1 \le j \le k} p_{i,j} \cdot \ClassValOpt{j} = \ClassValOpt{i}\). 
It is helpful to split this sum based on whether \(j \in B\) or \(j \in C\), i.e., whether \(1 \le j \le m\) or \(m + 1 \le j \le k\). 
We rewrite the sums as \(\sum_{j \in B} p_{i,j} \ClassValOpt{j} + \sum_{j \in C} p_{i,j} \ClassValOpt{j} = \ClassValOpt{i} \) for all \(i \in B\), and we represent this system of equations below using matrices.
\[
\begin{pmatrix}
    p_{1,1} & p_{1,2} & \cdots & p_{1,m} \\
    p_{2,1} & p_{2,2} & \cdots & p_{2,m} \\ 
    \vdots & \vdots & \ddots & \vdots \\
    p_{m,1} & p_{m,2} & \cdots & p_{m,m}
\end{pmatrix}
\begin{pmatrix}
    \ClassValOpt{1} \\
    \ClassValOpt{2} \\
    \vdots \\
    \ClassValOpt{m} 
\end{pmatrix}
+ 
\begin{pmatrix}
    p_{1,m+1} & p_{1,m+2} & \cdots & p_{1,k} \\
    p_{2,m+1} & p_{2,m+2} & \cdots & p_{2,k} \\
    \vdots & \vdots & \ddots & \vdots \\
    p_{m,m+1} & p_{m,m+2} & \cdots & p_{m,k} \\
\end{pmatrix}
\begin{pmatrix}
    \ClassValOpt{m+1} \\
    \ClassValOpt{m+2} \\
    \vdots \\
    \ClassValOpt{k} 
\end{pmatrix}
=
\begin{pmatrix}
    \ClassValOpt{1} \\
    \ClassValOpt{2} \\
    \vdots \\
    \ClassValOpt{m} 
\end{pmatrix}
\]
This system of equations is of the form \(Q_B \ClassValOpt{B} + Q_C \ClassValOpt{C} = \ClassValOpt{B}\). 
Rearranging terms gives us \((I - Q_{B}) \ClassValOpt{B} = Q_{C} \ClassValOpt{C} \) where \(I\) is the \(m \times m\) identity matrix.
\[
\begin{pmatrix}
    1 - p_{1,1} &   - p_{1,2} & \cdots &   - p_{1,m} \\
      - p_{2,1} & 1 - p_{2,2} & \cdots &   - p_{2,m} \\ 
    \vdots & \vdots & \ddots & \vdots \\
      - p_{m,1} &   - p_{m,2} & \cdots & 1 - p_{m,m}
\end{pmatrix}
\begin{pmatrix}
    \ClassValOpt{1} \\
    \ClassValOpt{2} \\
    \vdots \\
    \ClassValOpt{m} 
\end{pmatrix}
= 
\begin{pmatrix}
    p_{1,m+1} & p_{1,m+2} & \cdots & p_{1,k} \\
    p_{2,m+1} & p_{2,m+2} & \cdots & p_{2,k} \\
    \vdots & \vdots & \ddots & \vdots \\
    p_{m,m+1} & p_{m,m+2} & \cdots & p_{m,k} \\
\end{pmatrix}
\begin{pmatrix}
    \ClassValOpt{m+1} \\
    \ClassValOpt{m+2} \\
    \vdots \\
    \ClassValOpt{k} 
\end{pmatrix}
\]
The coefficients \(p_{i,j}\) depend on the specific choice of the representative \kl{boundary vertices} \(u_i\) for each \(\OptClass{i}\) for \(i \in B\).
However, since the weighted sum \(\sum_{v \in \OutNeighbours{u}} \ProbabilityFunction(u)(v) \cdot \ValOpt{v}\) is equal for all \kl{boundary vertices} \(u\) belonging to the same class \(\OptClass{i}\) by \ConditionBellman, we have that the solution of \(\ClassValOpt{1}, \ldots, \ClassValOpt{m}\) is independent of the chosen representatives. 
Hence, we are free to choose any representative \kl{boundary vertex} from each class.
It follows from \Cref{lem:d-invertible} that the equation \((I - Q_{B}) \ClassValOpt{B} = Q_{C} \ClassValOpt{C} \) has a unique solution.
\begin{restatable}{proposition}{DInvertible}
\label{lem:d-invertible}
The matrix \((I - Q_B)\) is invertible.
\end{restatable}
\begin{proof}
We first show that \(\lim_{n \to \infty} Q_B^n = 0\).
Then we show how this implies that \(I - Q_B\) is invertible.

We show \(\lim_{n \to \infty} Q_B^n = 0\) by showing that for every row of \(Q_B^m\) (where \(m = \abs{B}\)), the sum of the elements in the row is strictly less than \(1\). 
Recall that for \(1 \le i, j \le m\), we have that \(p_{i,j}\)  (the \(ij^{\text{th}}\) element of \(Q_B\))
 denotes the probability of reaching \(\OptClass{j}\) from \(u_i\) in one step, where \(u_i\) is the representative \kl{boundary vertex} chosen from \(\OptClass{i}\).
The matrix \(Q_B\) can be viewed as a transition probability matrix of a Markov chain with states \(\{M_1, M_2, \ldots, M_m\}\) (corresponding to the value classes \(\OptClass{1}, \ldots, \OptClass{m}\) respectively) and an additional sink state (corresponding to the value classes in \(C\)).
The probability of going from \(M_i\) to \(M_j\) in one step is \(p_{i,j}\) and the probability of going from \(M_i\) to the sink state in one step is \(1 - \sum_{j=1}^m {p_{i,j}}\).

By Proposition~\ref{prop:boundary-vertices-lower-higher}, we have that for every \(1 \le i \le m\), either the sum of the elements in the \(i^{\text{th}}\) row of \(Q_B\) is strictly less than 1, or there exist \(i_1, i_2\) such that \(1 \le i_1 < i <  i_2 \le m\) and \(p_{i,i_1} > 0\) and \(p_{i, i_2} > 0\). 
In particular, we have that the sums of elements in the \(1^{\text{st}}\) and \(m^{\text{th}}\) rows are both strictly less than 1.

For \(n \ge 1\), the \(ij^{\text{th}}\) element of \(Q_B^n\) denotes the probability of being in state \(M_j\) in the Markov chain starting from \(M_i\) after exactly \(n\) steps.
The sum of elements of the \(i^{\text{th}}\) row of \(Q_B\) is the probability that starting from \(M_i\), the Markov chain does not reach the sink state after \(n\) steps.
Since there are \(m\) non-sink states, there is a path of length at most \(m\) from each state to the sink state. 
Thus, the probability that the token reaches the sink state after \(m\) steps is positive.

The probability of every transition in the Markov chain is at least \(1/\ProbabilityDenominator\). 
Thus, the probability that the sink state is reached after \(m\) steps is at least \(1/\ProbabilityDenominator^m\). 
Thus, the sum of elements in each row of \(Q_B^m\) is at most \(1 - 1/\ProbabilityDenominator^{m}\).
For all \(n \ge 1\), the sum of elements in each row of \((Q_B^{m})^n\) is at most \((1 - 1/\ProbabilityDenominator^m)^n\). 
Thus, as \(n \to \infty\), the matrix \(Q_B^{mn} \to 0\).

Since \(\lim_{n \to \infty} Q_B^n =0\), we have that the absolute value of every eigenvalue of \(Q_B\) is strictly less than \(1\).
To see this, let \(\lambda\) be an eigenvalue of \(Q_B\) with eigenvector \(v\).
For all \(n \ge 0\), we have \(Q_B^n v = \lambda^n v\).
Since \(\lim_{n\to\infty} Q_B^n = 0\), the norm of \(Q_B^n v\) goes to \(0\) as \(n \to \infty\). 
The norm of \(\lambda^n v\) also goes to \(0\) as \(n \to \infty\), and  hence, we have \(\abs{\lambda} < 1\). 
In particular, \(1\) is not an eigenvalue of \(Q_B\), and hence, \(0\) is not an eigenvalue of \(I - Q_B\).
All eigenvalues of \(I- Q_B\) are non-zero, and since the determinant of a matrix is the product of the eigenvalues of the matrix, we have that the determinant of \(I - Q_B\) is non-zero.
\end{proof}

Let \(\alpha\) denote the least common multiple (lcm) of the denominators of \(p_{i,j}\) for \(1 \le i \le m\) and \(1 \le j \le k\).
We have \(0 < \alpha \le \ProbabilityDenominator^{mk}\), where \(\ProbabilityDenominator\) is the maximum denominator over all edge probabilities in \(\Game\).
We multiply both sides of the equation \((I - Q_B) \ClassValOpt{B} = Q_c \ClassValOpt{C}\) by \(\alpha\) to get \(\alpha (I - Q_B) \ClassValOpt{B} = \alpha Q_C \ClassValOpt{C}\) and note that all the elements of \(\alpha(I-Q_B)\) and \(\alpha Q_C\) are integers.
Let \(D\) be the determinant of the matrix \(\alpha(I-Q_B)\), and for \(1 \le i \le m\), let \(N_i\) be the determinant of the matrix obtained by replacing the \(i^{\text{th}}\) column of \(\alpha(I-Q_B)\) with the column vector \(\alpha Q_C \ClassValOpt{C}\). 
Since \(\alpha(I - Q_B)\) is invertible, by Cramer's rule~\cite{HK71},  we have that \(\ClassValOpt{i} = N_i/D\).
\Cref{lem:d-upper-bound} shows that \(\abs{D}\) is an integer and is at most \((2\alpha)^m\) and \Cref{lem:fwmpl-ni-denominator-upper-bound} shows that \(N_i\) has denominator at most \((\DenBoundNoBoundaryVertex{\Objective})^{k-m}\).
\begin{restatable}{proposition}{DUpperBound}%
\label{lem:d-upper-bound}
    The absolute value of the determinant of \(\alpha(I - Q_B)\), i.e., \(\abs{D}\), is an integer and is at most \((2\alpha)^m\), which is at most \(2^{\abs{\Vertices}} \cdot \ProbabilityDenominator^{\abs{V}^3}\).
\end{restatable}
\begin{proof}
    Every element of \(\alpha(I - Q_B)\) is an integer, and hence, \(D\) is an integer.
    To see the upper bound on \(\abs{D}\), observe some properties satisfied by \(\alpha(I - Q_B)\).
    \begin{enumerate}
        \item Each element in \(\alpha(I - Q_B)\) belongs to the set \(\{-\alpha, -\alpha + 1, \ldots, 0, \ldots, \alpha - 1, \alpha\}\). 
        This is because every element of \((I - Q_B)\) is between \(-1\) and \(1\), and because each element of \(\alpha(I - Q_B)\) is an integer.
        \item For each row of \(\alpha (I - Q_B)\), the sum of the absolute values of the elements in the row is at most \(2\alpha\). 
        To see this, note that the elements in \(Q_B\) are between \(0\) and \(1\) and each row in \(Q_B\) adds up to \(1\). Thus, the sum of absolute values of elements in a row of \(I - Q_B\) is at most \(2\).
    \end{enumerate}
    Note that the two properties hold for every minor of \(\alpha(I - Q_B)\) as well.
    
    We show that \(\abs{D}\) is at most \((2\alpha)^m\) by induction on the minors of \(\alpha(I - Q_B)\). 
    Consider any \(1 \times 1\) minor of \(\alpha(I - Q_B)\), that is, an element of \(\alpha(I - Q_B)\).
    The absolute value of such a minor is at most \(\alpha\), which satisfies the property of being at most \((2\alpha)^1\). 
    Now, consider a \(t \times t\) submatrix \(M\) of \(\alpha(I - Q_B)\) for \(1 < t \le m\). 
    By the induction hypothesis, every \((t-1) \times (t-1)\) minor of \(M\) has absolute value at most \((2\alpha)^{t-1}\). 
    Pick any row of \(M\) and compute the determinant of \(M\) by expanding along this row. 
    Since the sum of the absolute values of elements in this row is at most \(2\alpha\), the absolute value of the determinant of \(M\) is at most \((2\alpha) \cdot (2\alpha)^{t-1}\), which is \((2\alpha)^t\).
    Thus, the absolute value \(\abs{D}\) of the determinant \(D\) of \(\alpha(I - Q_B)\) is at most \((2\alpha)^m\). 
    Since \(\alpha \le \ProbabilityDenominator^{mk}\), \(m \le \abs{V}\), and \(k \le \abs{V}\), we have that \(\abs{D}\) is at most \(2^{\abs{V}} \cdot \ProbabilityDenominator^{\abs{V}^3}\). 
\end{proof}
\begin{restatable}{proposition}{FWMPLNiDenominatorUpperBound}%
\label{lem:fwmpl-ni-denominator-upper-bound}
    The denominator of \(N_i\) is at most  \((\DenBoundNoBoundaryVertex{\Objective})^{k-m}\), which is at most \((\DenBoundNoBoundaryVertex{\Objective})^{\abs{\Vertices}}\).
\end{restatable}
\begin{proof}
    Observe that all elements of the matrix \(\alpha Q_C\) are integers, and hence, \(\alpha Q_C\ClassValOpt{C}\) is a column vector where each element is a weighted sum of \(\ClassValOpt{m+1}, \ClassValOpt{m+2}, \ldots, \ClassValOpt{k}\) with integer coefficients.
    Note that all elements of \(\alpha(I - Q_B)\) are integers as well. 
    Thus, \(N_i\) is another weighted sum of the form
    \(a_{m+1} \ClassValOpt{m+1} + \cdots + a_k  \ClassValOpt{k}\) for some integer coefficients \(a_{m+1}, a_{m+2}, \ldots, a_{k}\).
    If written as a fraction in its reduced form, the denominator of \(N_i\) is at most the lcm of the denominators of  \(\ClassValOpt{m+1}, \ldots, \ClassValOpt{k}\).
    Since there are at most \((k-m)\) distinct elements in the set of denominators of \(\ClassValOpt{m+1}, \ldots, \ClassValOpt{k}\) and each denominator is at most \(\DenBoundNoBoundaryVertex{\Objective}\), we have that the lcm of this set is at most \((\DenBoundNoBoundaryVertex{\Objective})^{k-m}\).
    Thus, the denominator of \(N_i\) is at most \((\DenBoundNoBoundaryVertex{\Objective})^{k-m}\).
\end{proof}
Since the denominator of \(\ClassValOpt{i}\) is at most \(|D|\) times the denominator of \(N_i\), we obtain the bound stated in \Cref{thm:denominator-bound}.
\qed

\section{Expectation of window mean-payoff objectives}%
\label{sec:window-mean-payoff-objectives}
In this section, we apply the results from the previous section for two types of window mean-payoff objectives introduced in~\cite{CDRR15}: 
\begin{romanenumerate}
    \item \emph{fixed window mean-payoff} (\(\FWMPL\)) in which a window length \(\WindowLength \geq 1\) is given,  and 
    \item \emph{bounded window mean-payoff} (\(\BWMP\)) in which for every \kl{play}, we need a bound on window lengths. 
\end{romanenumerate}
We define these objectives below.

For a \kl{play} \(\Play\) in a \kl{stochastic game} \(\Game\), the \emph{mean payoff} of an infix \(\PlayInfix{i}{i+n} \)  is the average of the payoffs of the edges in the infix and  is defined as \(\mathsf{MP}(\PlayInfix{i}{i+n}) = \sum_{k=i}^{i+n-1} \frac{1}{n} \PayoffFunction(\Vertex[k], \Vertex[k+1])\). 
Given a window length \(\WindowLength \geq 1\) and a threshold \(\Threshold \in \Reals\), a \kl{play} \(\Play = \Vertex[0] \Vertex[1] \dotsm \)  in \(\Game\) satisfies the \emph{fixed window mean-payoff objective} \(\FWMP_{\Game}(\WindowLength, \Threshold)\) if from every position after some point, it is possible to start an infix of length at most $\WindowLength$ with mean payoff at least \(\Threshold\).
\begin{align*}
     \FWMP_{\Game}(\WindowLength, \Threshold) = \{ \Play \in \PlaySet{\Game} \suchthat  \exists k \geq 0 \cdot \forall i \ge k \cdot \exists j \in \PositiveSet{\WindowLength} \colon \mathsf{MP}(\Play(i, i + j)) \ge \Threshold\}
\end{align*} 
We omit the subscript \(\Game\) when it is clear from the context. 
We extend the definition of \emph{windows} as defined in~\cite{CDRR15} for arbitrary thresholds.
Given a threshold \(\Threshold\), a \kl{play} \(\Play = \Vertex[0] \Vertex[1] \cdots\), and \(0 \le i<  j\), we say that the \emph{\(\Threshold\)-window} \(\Play(i, j)\) is \emph{open} if the mean payoff of \(\Play(i,k)\) is less than \(\Threshold\) for all \(i < k \le j\). 
Otherwise, the \(\Threshold\)-window is \emph{closed}.
A \kl{play} \(\Play\) satisfies \(\FWMP(\WindowLength, \Threshold)\) if and only if from some point on, every \(\Threshold\)-window in \(\Play\) closes within at most \(\WindowLength\) steps.
Note that \(\FWMP(\WindowLength, \Threshold) \subseteq  \FWMP(\WindowLength', \Threshold) \) for \(\WindowLength \le \WindowLength'\) as a smaller window length is a stronger constraint.

We also consider another window mean-payoff objective called the \emph{bounded window mean-payoff objective} \(\BWMP[\Game](\Threshold)\).
A \kl{play} satisfies the objective \(\BWMP(\Threshold)\) if there exists a window length \(\WindowLength \ge 1\) such that the \kl{play} satisfies \(\FWMP(\WindowLength, \Threshold)\).
\[
    \BWMP_{\Game}(\Threshold) = \{ \Play \in \PlaySet{\Game} \suchthat \exists \WindowLength \ge 1 : \Play \in \FWMP_{\Game}(\WindowLength, \Threshold)\}
\]
Equivalently, a \kl{play} \(\Play\) does not satisfy \(\BWMP(\Threshold)\) if and only if for every suffix of \(\Play\), for all \(\WindowLength \ge 1\), the suffix contains an open \(\Threshold\)-window of length \(\WindowLength\).
Note that both \(\FWMP(\WindowLength, \Threshold)\) and \(\BWMP(\Threshold)\) are Boolean \kl{prefix-independent} objectives.

\subparagraph{Expected window mean-payoff values.}%
Corresponding to the \kl{Boolean objectives} \(\FWMP(\WindowLength, \Threshold)\) and \(\BWMP(\Threshold)\), we define quantitative versions of these objectives.
Given a \kl{play} \(\Play\) in  a \kl{stochastic game} \(\Game\) and a window length \(\WindowLength\), the \(\ObjectiveFWMPL\)-value of \(\Play\) is \(\sup \{\Threshold \in \Reals \suchthat \Play \in \FWMP_{\Game}(\WindowLength, \Threshold) \} \), the supremum threshold~\(\Threshold\) such that the \kl{play} satisfies \(\FWMP_{\Game}(\WindowLength, \Threshold)\). 
Using notations from \Cref{sec:preliminaries}, we denote the expected \(\ObjectiveFWMPL\)-value of a vertex \(v\) by \(\ExpValueFWMP{v}\).
We define \(\ExpValueBWMP{v}\), the expected \(\ObjectiveBWMP\)-value of a vertex \(v\) analogously.
If \(\ObjectiveBound{}\) is an integer such that the payoff \(\PayoffFunction(e)\) of each edge \(e\) in \(\Game\) satisfies \(|\PayoffFunction(e)| \le \ObjectiveBound{}\), then 
for all \kl{plays} \(\Play\) in \(\Game\), we have that \(\ObjectiveFWMPL(\Play)\) and \(\ObjectiveBWMP(\Play)\) lie between \(-\ObjectiveBound{}\) and \(\ObjectiveBound{}\).
Thus, \(\ObjectiveFWMPL\) and \(\ObjectiveBWMP\) are bounded objectives.

\subparagraph{Inductive property of windows.}
In~\cite{CDRR15}, the inductive property of \(\Threshold\)-windows has been defined for \(\Threshold = 0\). 
Here, we generalize this property to arbitrary values of \(\Threshold\). 
\begin{proposition}[Inductive property of \(\Threshold\)-windows]
\label{prop:inductive-property-of-windows}
    If the \(\Threshold\)-window starting at position \(i\) closes at position \(j\), then for all \(i \le k < j\), the \(\Threshold\)-windows \(\Play(k, j)\) are closed. 
\end{proposition}
The proof for arbitrary values of \(\Threshold\) as compared to value \(0\) as given in~\cite{CDRR15} is slightly more  involved and hence we present the proof here for completion.
\begin{proof}
    Since the \(\Threshold\)-window starting at position \(i\) closes at position \(j\), we have that for all \(i \le k < j\), \(\MP(\PlayInfix{i}{k}) < \Threshold\) and \(\MP(\PlayInfix{i}{j}) \ge \Threshold\).
    The mean payoff of the infix \(\PlayInfix{i}{j}\) is a weighted average (with positive weights) of the mean payoff of \(\PlayInfix{i}{k}\) and the mean payoff of \(\PlayInfix{k}{j}\). 
    Since \(\MP(\PlayInfix{i}{k}) < \Threshold\) and \(\MP(\PlayInfix{i}{j}) \ge \Threshold\), this implies that \(\MP(\PlayInfix{k}{j}) \ge \Threshold\). 
    Thus, the \(\Threshold\)-window starting at position \(k\) and ending at position \(j\) is closed. 
\end{proof}

\subparagraph{Equivalence of FWMP(1) and \Buchi\ objectives.}

Note that when \(\WindowLength = 1\), the \(\FWMP(1, \Threshold)\) and \(\overline{\FWMP(1, \Threshold)}\) (i.e., the complement of \(\FWMP(1, \Threshold)\)) objectives reduce to co\Buchi\ and \Buchi\ objectives respectively. 
To see this, let \(T\) be the set of all vertices \(v \in \Vertices\) such that either \(v \in \VerticesMain\) and all out-edges of \(v\) have payoff strictly less than \(\Threshold\), or \(v \in \VerticesAdversary\) and at least one out-edge of \(v\) has a payoff strictly less than \(\Threshold\). 
Then, a \kl{play} satisfies the \(\overline{\FWMP(1, \Threshold)}\) objective if and only if it satisfies the \textsf{\Buchi}\((T)\) objective. 

\begin{remark}
    We note that the optimal strategy to maximize the probability of getting nonnegative \(\ObjectiveFWMPL\)-value may be different from the optimal strategy to maximize the expected \(\ObjectiveFWMPL\)-value.
    Consider the game shown in \Cref{fig:swmp-vs-ewmp} with initial vertex \(v_1\).
    If the objective of \(\PlayerMain\) is to maximize the probability of satisfying \(\FWMP(\WindowLength, 0)\), then the optimal strategy is to move to \(v_2\) as that ensures a positive \(\ObjectiveFWMPL\)-value with probability \(0.9\).
    However, if the optimal strategy is to maximize the expected \(\ObjectiveFWMPL\)-value, then the optimal strategy is to move to \(v_5\). 
    \begin{figure}[t]
        \centering
        \begin{tikzpicture}
            \node[state] (v1) {\(\Vertex[1]\)};
            \node[random, draw, right of=v1] (v2) {\(\Vertex[2]\)};
            \node[state, right of=v2, yshift=6mm] (v3) {\(\Vertex[3]\)};
            \node[state, right of=v2, yshift=-6mm] (v4) {\(\Vertex[4]\)};
            \node[random, draw, left of=v1] (v5) {\(\Vertex[5]\)};
            \node[state, left of=v5, yshift=6mm] (v6) {\(\Vertex[6]\)};
            \node[state, left of=v5, yshift=-6mm] (v7) {\(\Vertex[7]\)};
            \draw 
                  (v1) edge[right, pos=0.3] node[below] {\(\EdgeValues{0}{}\)} (v2)
                  (v1) edge[right, pos=0.3] node[below] {\(\EdgeValues{0}{}\)} (v5)
                  (v2) edge[pos=0.6] node[above left] {\(\EdgeValues{0}{.9}\)} (v3)
                  (v2) edge[pos=0.6] node [below left] {\(\EdgeValues{0}{.1}\)} (v4)
                  (v3) edge[loop right, pos=0.5] node{\(\EdgeValues{+1}{}\)} (v3)
                  (v4) edge[loop right, pos=0.5] node{\(\EdgeValues{-10}{}\)} (v4)
                  (v5) edge[pos=0.6] node[above right] {\(\EdgeValues{0}{.1}\)} (v6)
                  (v5) edge[pos=0.6] node [below right] {\(\EdgeValues{0}{.9}\)} (v7)
                  (v6) edge[loop left, pos=0.5] node{\(\EdgeValues{+10}{}\)} (v6)
                  (v7) edge[loop left, pos=0.5] node{\(\EdgeValues{-1}{}\)} (v7)
            ;
        \end{tikzpicture}
        \caption{The optimal strategy for \(\PlayerMain\) from \(v_1\) is different in case of maximizing probability of \(\FWMP(\WindowLength, 0)\) and maximizing expected \(\ObjectiveFWMPL\)-value.}
        \label{fig:swmp-vs-ewmp}
    \end{figure}
    \lipicsEnd
\end{remark}

\begin{remark}
    In an MDP, for all vertices belonging to the same MEC, the expected \(\ObjectiveFWMPL\)-values of all the vertices are the same, and similarly, the expected \(\ObjectiveBWMP\)-values of all the vertices are also the same.
    However, this is not true for \kl{stochastic games} in general.
    Two vertices in the same MEC in a \kl{stochastic game} may have different expected \(\ObjectiveFWMPL\)-values and different expected \(\ObjectiveBWMP\)-values. 
    Consider the game shown in \Cref{fig:value-in-mecs-in-sg}. 
    All three vertices belong to the same MEC since each vertex is reachable from every other vertex.
    However, the expected \(\ObjectiveFWMPL\)-values of \(v_1\), \(v_2\), and \(v_3\) are \(-1\), \(0\), and \(+1\) respectively, for all values of \(\WindowLength\). 
    The expected \(\ObjectiveBWMP\)-values are also \(-1\), \(0\), and \(+1\) respectively. 
    \begin{figure}[t]
        \centering
        \begin{tikzpicture}
            \node[square, draw] (v1) {\(\Vertex[1]\)};
            \node[random, draw, right of=v1] (v2) {\(\Vertex[2]\)};
            \node[state, right of=v2] (v3) {\(\Vertex[3]\)};
            \draw 
                  (v1) edge[loop left, pos=0.5] node{\(\EdgeValues{-1}{}\)} (v1)
                  (v1) edge[bend right, pos=0.3] node[below] {\(\EdgeValues{0}{}\)} (v2)
                  (v2) edge[bend right, pos=0.3] node[above] {\(\EdgeValues{0}{.5}\)} (v1)
                  (v2) edge[bend left, pos=0.3] node [above] {\(\EdgeValues{0}{.5}\)} (v3)
                  (v3) edge[bend left, pos=0.3] node [below] {\(\EdgeValues{0}{}\)} (v2)
                  (v3) edge[loop right, pos=0.5] node{\(\EdgeValues{+1}{}\)} (v3)
            ;
        \end{tikzpicture}
        \caption{Different vertices in the same MEC in a \kl{stochastic game} may have different expected \(\ObjectiveFWMPL\)-values.}
        \label{fig:value-in-mecs-in-sg}
    \end{figure}
    \lipicsEnd
\end{remark}

\begin{remark}[Achieving supremum for expected \(\ObjectiveFWMPL\) and \(\ObjectiveBWMP\) objectives.]
The \(\ObjectiveFWMPL\)-value of a \kl{play} is the supremum \(\Threshold\) such that every \(\Threshold\)-window is closed in at most \(\WindowLength\) steps. Since there are finitely many sequences of edge payoffs of length at most \(\WindowLength\), the supremum is achieved for the \(\ObjectiveFWMPL\) objective for all \kl{plays} \(\Play\), and thus can be replaced by max.
The supremum for \(\ObjectiveBWMP\)-value may not be reached in general, as Player 2 can keep the window open for increasing window lengths where the supremum is approached asymptotically.
    \lipicsEnd
\end{remark}

\subparagraph{Decision problems.}
Given a \kl{stochastic game} \(\Game\), a vertex \(v\), and a threshold \(\Threshold \in \Rationals\), we have the following expectation problems for the window mean-payoff objectives:
\begin{itemize}
    \item \emph{expected \(\ObjectiveFWMPL\)-value problem}: 
    Given a window length \(\WindowLength \ge 1\), 
    is \(\ExpValueFWMP{v} \ge \Threshold\)?
    \item \emph{expected \(\ObjectiveBWMP\)-value problem}: 
    Is \(\ExpValueBWMP{v} \ge \Threshold\)?
\end{itemize}
As considered in previous works~\cite{CDRR15,BGR19,BDOR20}, the window length $\WindowLength$ is usually small ($\WindowLength \leq \abs{V}$), and hence we assume that $\WindowLength$ is given in unary (while the edge payoffs are given in binary).

\subsection{Expected fixed window mean-payoff value}%
\label{sec:expected-fixed-window-mean-payoff-value}
We give tight complexity bounds for the expected \(\ObjectiveFWMPL\)-value problem.
We use the characterization from \Cref{thm:six-conditions} to present our main result that this problem is in \(\UP \intersection \coUP\)~(\Cref{thm:fwmp-summary}).
First, we give some preliminary complexity lower and upper bounds in \Cref{lem:ssg-hardness-np-conp} and \Cref{lem:NEXPcoNEXP} respectively.

\subparagraph{Reduction to simple stochastic games.}
We show in \Cref{lem:ssg-hardness-np-conp} that \kl{simple stochastic games}~\cite{Con92},
which are in \(\UP \intersection \coUP\)~\cite{CF11}, reduce to the expected \(\ObjectiveFWMPL\)-value problem, giving a tight lower bound.
\begin{lemma}
\label{lem:ssg-hardness-np-conp}
    The expected \(\ObjectiveFWMPL\)-value problem is at least as hard as \kl{simple stochastic games}.
\end{lemma}
\begin{proof}
The reduction goes as follows. 
Recall that in a \kl{simple stochastic game}, \(\PlayerMain\) wins from a vertex \(v\) if and only if she has a strategy that ensures that starting from \(v\), the probability that the token eventually reaches the target vertex \(v_{target}\) is greater than \(\frac{1}{2}\).
We assume without loss of generality that \(v_{target}\) is absorbing, that is, the only out-neighbour of \(v_{target}\) is \(v_{target}\) itself.

Given a \kl{simple stochastic game} \(\Game_{SSG}\), we construct a new \kl{stochastic game} \(\Game\) such that \(\PlayerMain\) reaches \(v_{target}\) in \(\Game_{SSG}\) from \(v\) with probability greater than \(\frac{1}{2}\) if and only if the expected \(\ObjectiveFWMPL\)-value of \(v\) in \(\Game\) is greater than \(\frac{1}{2}\).  

The set of vertices and edges in \(\Game\) are the same as in \(\Game_{SSG}\).
We let the edge payoff of the self-loop of \(v_{target}\) be \(1\) in \(\Game\), and let the edge payoff of every other edge in \(\Game\) be \(0\). 
Thus, the probability of reaching \(v_{target}\) from \(v\) in \(\Game_{SSG}\) is equal to the expected \(\ObjectiveFWMPL\)-value of \(v\) in \(\Game\).
Hence, the expected \(\ObjectiveFWMPL\)-value problem is at least as hard as \kl{simple stochastic games}, which are known to be in \(\UP \intersection \coUP\).
\end{proof}
We note that the proof above is independent of the exact value of the window length \(\WindowLength\).

\subparagraph{\(\NEXP \intersection \coNEXP\) upper bound.} 
The expected \(\ObjectiveFWMPL\)-value problem can be reduced to the expected \(\liminfObj\)-value problem~\cite{CH09} on an exponentially larger game.
This gives an \(\NEXP \intersection \coNEXP\) algorithm for expected \(\ObjectiveFWMPL\) since the expected \(\liminfObj\)-value problem for \kl{stochastic games} is in \(\NP \intersection \coNP\)~\cite{CH09}. 
\begin{lemma}
\label{lem:NEXPcoNEXP}
    The expected \(\ObjectiveFWMPL\)-value problem is in \(\NEXP \intersection \coNEXP\). 
\end{lemma}
\begin{proof}
Starting with a \kl{stochastic game} \(\Game = ((\Vertices, \Edges), (\VerticesMain, \VerticesAdversary, \VerticesRandom), \ProbabilityFunction, \PayoffFunction)\), we construct an exponentially larger \kl{stochastic game} \(\Game' = ((\Vertices', \Edges'), (\VerticesMain', \VerticesAdversary', \VerticesRandom'), \ProbabilityFunction', \PayoffFunction')\) such that the expected \(\ObjectiveFWMPL\)-value of the initial vertex in \(\Game\) is at least \(\Threshold\) if and only if the expected \(\liminfObj\)-value of the initial vertex in the constructed game \(\Game'\) is at least \(\Threshold\). 

Intuitively, each vertex in \(\Game'\) is a history of the last \(\WindowLength\) vertices seen in \(\Game\), and thus, the game \(\Game'\) has exponentially many vertices than \(\Game\). 
The value of a \kl{play} in \(\Game\) depends on the payoffs seen in a sliding window of length \(\WindowLength\).
Since each vertex in \(\Game'\) stores the \(\WindowLength\)-length history, the value of a \kl{play} in \(\Game'\) can be described simply in terms of a \(\liminfObj\) objective.

We construct \(\Game'\) as follows. 
The set \(\Vertices'\) of vertices in \(\Game'\) is equal to \(\Vertices^{\WindowLength + 1}\), and thus, there are \(\abs{\Vertices}^{\WindowLength + 1}\) vertices in \(\Game'\). 
We have a label of length \(\WindowLength + 1\) vertices for each vertex in \(\Vertices'\). 
For all \(1 \le i \le \WindowLength + 1\) and \(v' \in \Vertices'\), let \(\Label{i}{v'}\) denote the \(i^{\text{th}}\) coordinate of \(v'\). 
That is, \(v' = (\Label{1}{v'}, \Label{2}{v'}, \ldots, \Label{\WindowLength + 1}{v'})\). 
Each vertex \(v' \in \Vertices'\) belongs to \(\VerticesMain'\), \(\VerticesAdversary'\), or \(\VerticesRandom'\) depending on whether the last coordinate of \(v'\) (i.e., \(\Label{\WindowLength+1}{v'}\)) belongs to \(\VerticesMain\), \(\VerticesAdversary\), or \(\VerticesRandom\). 
Formally, we have that \(v' \in \VerticesMain'\) if \(\Label{\WindowLength + 1}{v'} \in \VerticesMain\),
and \(v' \in \VerticesAdversary'\) if \(\Label{\WindowLength + 1}{v'} \in \VerticesAdversary\),
and \(v' \in \VerticesRandom'\) if \(\Label{\WindowLength + 1}{v'} \in \VerticesRandom\).

Next, we describe the edges in \(\Game'\). 
If \(v'\) is a vertex in \(\Game'\), then for all \(u \in \Vertices\) such that \(u\) is an out-neighbour of \(\Label{\WindowLength+1}{v'}\) in \(\Game\), we have that \((\Label{2}{v'}, \Label{3}{v'}, \ldots, \Label{\WindowLength + 1}{v'}, u)\) is an out-neighbour of \(v'\) in \(\Game'\). 
Formally, \(\Edges'(v') = \{ (\Label{2}{v'}, \Label{3}{v'}, \ldots, \Label{\WindowLength +1}{v'}, u) \suchthat u \in \OutNeighbours{\Label{\WindowLength + 1}{v'}}\} \).
Intuitively, we pop out the vertex in the first coordinate and  push in a new vertex in the last coordinate. 
Note that the degree of vertex \(v'\) in \(\Game'\) is equal to the degree of \(\Label{\WindowLength + 1}{v'}\) in \(\Game\). 

Now, we define the probability distribution \(\ProbabilityFunction'(v')\) over probabilistic vertices in \(\Game'\). 
If \(v'\) is a probabilistic vertex, i.e., if \(v' \in \VerticesRandom'\), then for all out-neighbours \(u' \in \Edges'(v')\) of \(v'\), we have that \(\ProbabilityFunction'(v')(u') = \ProbabilityFunction(\Label{\WindowLength+1}{v'})( \Label{\WindowLength + 1}{u'})\). 

Finally, we define the payoffs \(\PayoffFunction'(e')\) of edges \(e'\) in \(\Game'\).
All edges coming out of the same vertex in \(\Game'\) are given the same payoff. 
That is, the edge payoffs of \((v', u_1')\) and \((v', u_2')\) are equal.
The edge payoff of edges out of \(v'\) is determined by the label \(\Label{1}{v'} \Label{2}{v'} \cdots \Label{\WindowLength + 1}{v'}\) of \(v'\).
The payoff depends on whether this label is a sequence of edges in \(\Game\), that is, if \((L_{i}(v'), L_{i+1}(v')) \in E\) for all \(1 \le i \le \WindowLength\). 
For all \(v' \in \Vertices'\), 
if the label of \(v'\) is a sequence of edges in \(\Game\), then for all edges \(e'\) out of \(v'\), i.e., for all \(e' \in \Edges'(v')\), the payoff \(\PayoffFunction'(e')\) of \(e'\) is equal to the maximum \(\Threshold\) such that the \(\Threshold\)-window starting at \(\Label{1}{v'}\) is closed at or before the end of the label. 
Otherwise, if it is not a sequence of edges in \(\Game\), then for all edges \(e'\) out of \(v'\), we define the payoff of \(\PayoffFunction'(e')\) to be \(0\).

If the initial vertex in \(\Game\) is \(v \in \Vertices\), then let the initial vertex in \(\Game'\) be \((v, v, \ldots, v)\), where the tuple has length \(\WindowLength + 1\) as stated above.
Observe that there is a one-to-one correspondence between \kl{plays} in \(\Game\) and \kl{plays} in \(\Game'\).
Starting with a \kl{play} \(\Play\) in \(\Game\), we show how to obtain the corresponding \kl{play} \(\Play'\) in \(\Game'\).
We start with \(v\) in \(\Play\) and \((v, v, \ldots, v)\) in \(\Play'\).
Each time we read a vertex, say \(u\), in \(\Play\), the next vertex in \(\Play'\) can be obtained by considering the label of the last vertex in \(\Play'\) until before \(u\) is read in \(\Play\), popping its first coordinate and pushing the new vertex \(u\) in its last coordinate.
As a consequence, after the first \(\WindowLength\) steps in \(\Game'\), every vertex visited in \(\Game'\) has a label that is a sequence of the last \(\WindowLength\) edges in \(\Game\). 
Conversely, given a \kl{play} \(\Play'\) in \(\Game'\), one can project \(\Play'\) to the last component to obtain the corresponding \kl{play} \(\Play\) in \(\Game\). 

If \(\Threshold\) is such that, eventually, every \(\Threshold\)-window in \(\Play\) closes in at most \(\WindowLength\) steps, then there are infinitely many edges in \(\Play'\) with payoff at least \(\Threshold\), in which case the \(\liminfObj\)-value of \(\Play'\) is at least \(\Threshold\).
Conversely, if \(\Threshold\) is such that there are infinitely many open \(\Threshold\)-windows of length \(\WindowLength\) in \(\Play\), then there are infinitely many edges in \(\Play'\) with payoff less than \(\Threshold\), and the \(\liminfObj\)-value of \(\Play'\) is less than \(\Threshold\). 
The correspondence between \kl{plays} in the two games, gives a correspondence between strategies between the games. 
The probability functions \(\ProbabilityFunction\) and \(\ProbabilityFunction'\) of \(\Game\) and \(\Game'\) respectively are defined such that once we fix strategies of the players in \(\Game\), and the corresponding strategies in \(\Game'\), then the probability distribution of sets of \kl{plays} in \(\Game\) is equal to the probability distribution of the corresponding sets of \kl{plays} in \(\Game'\).
Thus, the expected \(\ObjectiveFWMPL\)-value of a vertex \(v\) in \(\Game\) is equal to the expected \(\liminfObj\)-value of \((v,v, \ldots, v)\) in \(\Game'\). 

Since the expected \(\liminfObj\)-value problem is in \(\NP \intersection \coNP\)~\cite{CH09} and the size of \(\Game'\) is exponential in the size of \(\Game\), we have that the expected \(\ObjectiveFWMPL\)-value problem is in \(\NEXP \intersection \coNEXP\).
\end{proof}

In order to use the characterization from \Cref{thm:six-conditions}, we show the existence of the bound \(\DenBoundNoBoundaryVertex{\FWMPL}\) for the \(\ObjectiveFWMPL\) objective.
We show in \Cref{lem:fwmpl-values-without-boundary-vertices} that the expected \(\ObjectiveFWMPL\)-value \(\ClassValFWMP{i}\) of a class \(\FWMPClass{i}\) without \kl{boundary vertices} takes a special form, that is, \(\ClassValFWMP{i}\) is the mean payoff of a sequence of at most \(\WindowLength\) edges in \(\FWMPClass{i}\).
We use the fact that the \(\ObjectiveFWMPL\)-value of every \kl{play} \(\Play\) is the largest \(\Threshold\) such that, eventually, every \(\Threshold\)-window in \(\Play\) closes in at most \(\WindowLength\) steps, and that \(\Threshold\) is the mean payoff of a sequence of at most \(\WindowLength\) edges in \(\Play\).
To complete the argument, we show that if both players play optimally, then, with probability~\(1\), the \(\ObjectiveFWMPL\)-value of the outcome \(\Play\) is equal to \(\ClassValFWMP{i}\) and thus, \(\ClassValFWMP{i}\) is also of this form.
\begin{restatable}{lemma}{FWMPLValuesWithoutBoundaryVertices}%
\label{lem:fwmpl-values-without-boundary-vertices}
    The expected \(\ObjectiveFWMPL\)-value \(\ClassValFWMP{i}\) of vertices in a class \(\FWMPClass{i}\) without \kl{boundary vertices} is equal to the mean payoff of some sequence of \(\WindowLength\) or fewer edges in \(\FWMPClass{i}\). 
    That is, \(\ClassValFWMP{i}\) is of the form \(\frac{1}{j} \left(\PayoffFunction(e_1) + \cdots + \PayoffFunction(e_j) \right)\) for some \(j \le \WindowLength\) and edges \(e_1, e_2, \ldots, e_j\). 
\end{restatable}
\begin{proof}
    Let \(\Strategy[\Main]^*\) and \(\Strategy[\Adversary]^*\) be optimal strategies of \(\PlayerMain\) and \(\PlayerAdversary\) respectively for the expected \(\ObjectiveFWMPL\) objective in \(\FWMPClass{i}\) from an initial vertex \(v_0\) in \(\FWMPClass{i}\). 
    By definition, we have that the expected \(\ObjectiveFWMPL\)-value of the outcome of this strategy profile is equal to \(\ClassValFWMP{i}\). 
    We can in fact show a stronger statement: starting from \(v_0\), with probability~\(1\), the \(\ObjectiveFWMPL\)-value of the outcome of this strategy profile is equal to \(\ClassValFWMP{i}\). 
    If not, then with positive probability, the \(\ObjectiveFWMPL\)-value of the outcome is greater than \(\ClassValFWMP{i}\) and with positive probability, the \(\ObjectiveFWMPL\)-value of the outcome is less than \(\ClassValFWMP{i}\). 
    Since \(\ObjectiveFWMPL\) is a \kl{prefix-independent} objective, it follows that there exists another vertex \(v'\) from which \(\PlayerMain\) can ensure with probability~\(1\) that the \(\ObjectiveFWMPL\)-value of the outcome is greater than \(\ClassValFWMP{i}\). 
    It follows that the expected \(\ObjectiveFWMPL\)-value of \(v'\) is greater than \(\ClassValFWMP{i}\), which contradicts our hypothesis that the expected \(\ObjectiveFWMPL\)-value of each vertex in \(\FWMPClass{i}\) is exactly \(\ClassValFWMP{i}\).

    Now, we show that \(\ClassValFWMP{i}\) takes the form mentioned in the statement of \Cref{lem:fwmpl-values-without-boundary-vertices}. 
    Since \(\ObjectiveFWMPL\) is a \kl{prefix-independent} objective, the \(\ObjectiveFWMPL\)-value of the outcome \(\Play\) only depends only on \(\inf(\Play)\), the set of vertices visited infinitely often in \(\Play\). 
    For every probabilistic vertex \(v\) in \(\inf(\Play)\), each out-edge of \(v\) is also chosen infinitely often in \(\Play\). 
    Thus, the \(\ObjectiveFWMPL\)-value of the outcome is equal to supremum \(\Threshold\) such that every \(\Threshold\)-window closes in at most \(\WindowLength\) steps.
    This is the mean payoff of a sequence of at most \(\WindowLength\) edges that appears infinitely often in \(\Play\). 
    This gives that \(\ClassValFWMP{i}\) is equal to the mean payoff of some \(j\) edges in \(\FWMPClass{i}\) for some \(1 \le j \le \WindowLength\). 
\end{proof}

This observation gives us the bound \(\DenBoundNoBoundaryVertex{\FWMPL}\) on the denominators of the values of \(\ValVectorFWMP\)-classes without \kl{boundary vertices}.
To see this, let \(\PayoffDenominator = \max \{q \suchthat \exists p, q \in \Integers, \, \exists \Edge \in \Edges :  \PayoffFunction(e) = \frac{p}{q} \text{ with } p, q \text{ co-prime}\}\) be the maximum denominator over all edge payoffs in \(\Game\).
Since \(j \le \WindowLength\), and each \(\PayoffFunction(e_1), \PayoffFunction(e_2), \ldots, \PayoffFunction(e_j)\) is a rational number with denominator at most \(\PayoffDenominator\),  the denominator of the sum \(\PayoffFunction(e_1) + \cdots + \PayoffFunction(e_j)\) is at most \(\PayoffDenominator \cdot (\PayoffDenominator - 1) \cdot (\PayoffDenominator - 2) \cdots (\PayoffDenominator - (\WindowLength-1))\) if \(\PayoffDenominator \ge \WindowLength\), and at most \(\PayoffDenominator!\) if \(\PayoffDenominator \le \WindowLength\).
In both cases, this is at most \(\PayoffDenominator^{\WindowLength}\).

\begin{corollary}%
\label{cor:fwmpl-denominator-without-boundary-vertices}
    The expected \(\ObjectiveFWMPL\)-value of vertices in \(\ValVectorFWMP\)-classes without \kl{boundary vertices} can be written as \(\frac{p}{q}\) where \(p\) and \(q\) are integers and \(q \le \PayoffDenominator^{\WindowLength} \cdot \WindowLength\).
\end{corollary}
From \Cref{thm:denominator-bound}, we get that the denominator of \(\ClassValFWMP{i}\) for each class \(\FWMPClass{i}\) in \(\Game\) is at most \(2^{\abs{V}} \cdot \ProbabilityDenominator^{\abs{V}^3} \cdot (\DenBoundNoBoundaryVertex{\FWMPL})^{\abs{\Vertices}}\), which is at most \(2^{\abs{V}} \cdot \ProbabilityDenominator^{\abs{V}^3} \cdot (\PayoffDenominator^{\WindowLength} \cdot \WindowLength)^{\abs{\Vertices}}\).
\begin{restatable}{lemma}{FWMPLValueDenominator}
    \label{lem:fwmpl-value-denominator}
    The expected \(\ObjectiveFWMPL\)-value of each vertex in \(\Game\) can be written as a fraction \(\frac{p}{q}\), where \(p, q\) are integers, and \(q \le 2^{\abs{V}} \cdot \ProbabilityDenominator^{\abs{V}^3} \cdot (\PayoffDenominator^{\WindowLength} \cdot \WindowLength)^{\abs{V}}\), and \(-\ObjectiveBound{} \cdot q \le p \le \ObjectiveBound{} \cdot q\). 
\end{restatable}
We now state the main result of this section for the expected \(\ObjectiveFWMPL\)-value problem.

\begin{theorem}%
\label{thm:fwmp-summary}
    The expected \(\ObjectiveFWMPL\)-value problem is in \(\UP \intersection \coUP\) when \(\WindowLength\) is given in unary.
    Memory of size \(\WindowLength\) suffices for \(\PlayerMain\), while memory of size \(\abs{\Vertices} \cdot \WindowLength\) suffices for \(\PlayerAdversary\).
\end{theorem}
\begin{proof}
    To show membership of the expected \(\ObjectiveFWMPL\)-value problem in \(\UP \intersection \coUP\),
    we first guess the expected \(\ObjectiveFWMPL\)-value vector \(\ValVectorFWMP\), that is, the expected \(\ObjectiveFWMPL\)-value \(\ValFWMP{v}\) of every vertex \(v\) in the game.
    From \Cref{lem:fwmpl-value-denominator}, it follows that the number of bits required to write \(\ValFWMP{v}\) for every vertex \(v\) is polynomial in the size of the input.
    Thus, the vector \(\ValVectorFWMP\) can be guessed in polynomial time.

    Then, to verify the guess, it is sufficient to verify the \ConditionBellman, \ConditionLB, and \ConditionUB\ conditions for \(\ObjectiveFWMPL\).
    It is easy to see that the \ConditionBellman\ condition can be checked in polynomial time.
    Checking the \ConditionLB\ and \ConditionUB\ conditions, i.e., checking the almost-sure satisfaction of the \kl{threshold Boolean objective} \(\FWMP(\WindowLength, \Threshold)\) for appropriate thresholds \(\Threshold\) in \kl{trap subgames} in each \(\ValVectorFWMP\)-class can be done in polynomial time~\cite{DGG25LMCS}.
    Thus, the decision problem of \(\ExpValue{v}{\ObjectiveFWMPL} \ge \Threshold\) is in \(\NP\), 
    and moreover, since there is exactly one value vector that satisfies the conditions in \Cref{thm:six-conditions}, the decision problem is, in fact, in \(\UP\). 
    Analogously, the complement decision problem of \(\ExpValue{v}{\ObjectiveFWMPL} < \Threshold\) is also in \(\UP\).
    Hence, the expected \(\ObjectiveFWMPL\)-value problem is in \(\UP \intersection \coUP\). 
    
    From the description of the optimal strategy in \Cref{lem:fwmpl-almost-sure-to-expected}, it follows from \cref{cor:memory-bound} that the memory requirement for the expected \(\ObjectiveFWMPL\) objective is no greater than the memory requirement for the almost-sure satisfaction of the corresponding \kl{threshold objectives}, which are \(\WindowLength\) and \(\abs{\Vertices} \cdot \WindowLength\) for \(\PlayerMain\) and \(\PlayerAdversary\) respectively~\cite{DGG25LMCS}.
\end{proof}

\begin{example} \label{exm:FWMP}
    We show that the value vector \(\ValVector\) in \Cref{ex:vc-notations} for the game in \Cref{fig:swmp-example} satisfies all the conditions of~\Cref{thm:six-conditions} for the \(\ObjectiveFWMPL\) objective with \(\WindowLength = 2\).
    The \ConditionBellman\ condition can be verified by analysing the game graph.
    To check the \ConditionLB\ condition, we check if \(\PlayerMain\) satisfies the \kl{threshold Boolean objectives} \(\{\ObjectiveFWMPL > \ClassVal{i} - \GranularityBound{\FWMP}\}\) almost surely from every vertex in each \(\TrapSubgame{\Main}{\ValClass{i}}\) (computed in \Cref{ex:apt}).
    Since \(\TrapSubgame{\Main}{\ValClass{2}} = \emptyset\), the condition holds vacuously.
    For \(i \in \{1, 3, 5\}\), the game \(\TrapSubgame{\Main}{\ValClass{i}}\) consists of only one vertex with a self-loop of value \(\ClassVal{i}\), and thus the \kl{threshold objective} is satisfied almost surely.
    Similarly, for \(\TrapSubgame{\Main}{\ValClass{4}}\), there is only one \kl{play} where payoff \(0\) and \(2\) occur alternatingly, and it has \(\ObjectiveFWMPL\)-value equal to \(1\) for \(\WindowLength = 2\).
    The \ConditionUB\ condition can be verified analogously.
    Thus, we have that the vector \(\ValVector\) is equal to  the expected \(\ObjectiveFWMPL\)-value vector \(\ValVectorFWMP\).
    \lipicsEnd
\end{example}

\subsection{Expected bounded window mean-payoff value}%
\label{sec:expected-bounded-window-mean-payoff-value}

We would like to apply the characterization in \Cref{thm:six-conditions} to \(\ObjectiveBWMP\) to show that the expected \(\ObjectiveBWMP\)-value problem is in \(\UP \intersection \coUP\), and thus, we show the existence of the bound \(\DenBound{\BWMP}\) for the \(\ObjectiveBWMP\) objective.
We show in \Cref{lem:bwmp-values-without-boundary-vertices} that the expected \(\ObjectiveBWMP\)-value \(\ClassValBWMP{i}\) of a class \(\BWMPClass{i}\) without \kl{boundary vertices} is the mean payoff of a simple cycle in \(\BWMPClass{i}\).
While \Cref{lem:bwmp-values-without-boundary-vertices} is analogous to \Cref{lem:fwmpl-values-without-boundary-vertices} for \(\ObjectiveFWMPL\), the proof of \Cref{lem:bwmp-values-without-boundary-vertices} is more involved since the \(\ObjectiveBWMP\) objective requires one to consider windows of arbitrary lengths.
In the proof, we make use of the fact that memoryless strategies suffice for \(\PlayerMain\) to \kl{play} optimally for the almost-sure satisfaction of the \(\BWMP\) objective~\cite{DGG25LMCS}.
In the resulting MDP (which has the same set of vertices as the game \(\Game_{\BWMPClass{i}}\)), we carefully analyse the resulting \kl{plays} when \(\PlayerAdversary\) plays optimally.
\begin{restatable}{lemma}{BWMPValuesWithoutBoundaryVertices}
\label{lem:bwmp-values-without-boundary-vertices}
    The expected \(\ObjectiveBWMP\)-value \(\ClassValBWMP{i}\) of vertices in a class \(\BWMPClass{i}\) without \kl{boundary vertices} is equal to the mean-payoff value of a simple cycle in \(\BWMPClass{i}\). 
    That is, \(\ClassValBWMP{i}\) is of the form \(\frac{1}{j}(\PayoffFunction(e_1) + \cdots + \PayoffFunction(e_j))\) for some \(j \le \abs{\Vertices}\) and edges \(e_1, e_2, \ldots, e_j\) of a simple cycle.
\end{restatable}
\begin{proof}
    Since \(\BWMPClass{i}\) is a class without \kl{boundary vertices}, the game \(\Game_{\BWMPClass{i}}\) is  obtained by simply restricting \(\Game\) to the class \(\BWMPClass{i}\).
    The expected \(\ObjectiveBWMP\)-value of each vertex in \(\BWMPClass{i}\) in the game \(\Game\) is equal to \(\ClassValBWMP{i}\), which is also equal to the expected \(\ObjectiveBWMP\)-value of each vertex in the restriction game \(\Game_{\BWMPClass{i}}\).
    Thus, it is sufficient to consider the game \(\Game_{\BWMPClass{i}}\).
    Recall from the proof of \Cref{lem:fwmpl-almost-sure-to-expected} that an optimal strategy of \(\PlayerMain\) for the expected \(\ObjectiveBWMP\)-objective is to follow an optimal strategy for the almost-sure satisfaction of the \(\BWMP\) objective, for which it is known that  memoryless strategies suffice~\cite{DGG25LMCS}.
    We show a useful claim that holds for all memoryless strategies \(\Strategy[\Main]\) of \(\PlayerMain\).
    In particular, the claim also holds for optimal memoryless strategies of \(\PlayerMain\) for which the expected \(\ObjectiveBWMP\)-value \(\ClassValBWMP{i}\) is attained.

    Let \(\Strategy[\Main]\) be a memoryless strategy of \(\PlayerMain\) in \(\Game_{\BWMPClass{i}}\).
    Fixing the strategy \(\Strategy[\Main]\) in the game \(\Game_{\BWMPClass{i}}\) gives an MDP \(\OutcomeMDP{\Game_{\BWMPClass{i}}}{\Strategy[\Main]}\) with the same set of vertices as \(\Game_{\BWMPClass{i}}\).
    For ease of notation, we denote this MDP by \(\MDP\).
    The MDP \(\MDP\) can be decomposed into maximal end-components (MECs).
    For every MEC \(T\) in \(\MDP\), let \(\gamma_{T}\) denote the mean payoff of the simple cycle in \(T\) with the minimum mean payoff.
    The claim is the following:
    for every MEC \(T\) in \(\MDP\), in the MDP \(\MDP_{T}\) obtained by restricting \(\MDP\) to \(T\), the expected \(\ObjectiveBWMP\)-value of every vertex is the same and is equal to \(\gamma_{T}\).
    In other words, if the token reaches a MEC \(T\) in the MDP \(\MDP\) and \(\PlayerAdversary\) chooses to always stay in \(T\) and never exit, then the expected \(\ObjectiveBWMP\)-value of the outcome is equal to \(\gamma_{T}\).
    We prove this claim later.

    Now, we show that this observation implies that \(\ClassValBWMP{i}\) is equal to the mean payoff of a simple cycle in \(\BWMPClass{i}\).
    Recall that for all strategies of \(\PlayerAdversary\), the outcome in \(\MDP\) almost surely ends up in an MEC of \(\MDP\) from which it never exits.
    Moreover, since \(\ObjectiveBWMP\) is a \kl{prefix-independent} objective, the \(\ObjectiveBWMP\)-value of a \kl{play} only depends on the MEC that the \kl{play} ends up in.
    Thus, to obtain the expected \(\ObjectiveBWMP\)-value of vertices in the MDP \(\MDP\), we can \emph{collapse} each MEC \(M\), that is, we replace each MEC \(M\) with a single vertex \(v_M\). 
    The out-edges of \(v_M\) are the union of the out-edges of vertices in \(M\) to vertices not in \(M\), 
    The in-edges of \(v_M\) are the union of the in-edges of vertices in \(M\) from vertices not in \(M\). 
    In addition, we have a self-loop on \(v_M\) with payoff \(\gamma_{T}\). 
    The resulting MDP with the collapsed MECs is known as the MEC quotient of \(\MDP\)~\cite{BCCF14,ACDKM17,HM18,BGR19}.
    The expected \(\ObjectiveBWMP\)-value of vertices in the MDP \(\MDP\) is equal to the expected mean-payoff value of the corresponding vertices in the MDP with the collapsed MECs.

    In particular, when the strategy \(\Strategy[\Main]\) of \(\PlayerMain\) is an optimal strategy, we get that every vertex in \(\MDP\) has the same expected \(\ObjectiveBWMP\)-value. 
    That is, for each MEC \(T\) in \(\MDP\), we have that \(\gamma_{T}\) is at least \(\ClassValBWMP{i}\), and moreover, from each vertex in \(\MDP\), \(\PlayerAdversary\) has a strategy to almost surely eventually reach a MEC with value \(\ClassValBWMP{i}\).
    Thus, we get that \(\ClassValBWMP{i}\) is equal to the mean payoff of a simple cycle in \(\MDP\). 
    Since every simple cycle in \(\MDP\) is also a simple cycle in \(\BWMPClass{i}\), we have that \(\ClassValBWMP{i}\) is equal to the mean payoff of a simple cycle in \(\BWMPClass{i}\).
    
    \subparagraph{Proof of claim.}
    It remains to prove our claim. 
    Recall that for a MEC \(T\) in the MDP \(\MDP\), we denote by \(\MDP_{T}\) the restriction of \(\MDP\) to \(T\), and  we want to show that for all MECs \(T\) in \(\MDP\), the expected \(\ObjectiveBWMP\)-value of every vertex in \(\MDP_{T}\) is equal to \(\gamma_{T}\).
    First, we show that the expected \(\ObjectiveBWMP\)-value of every vertex in \(\MDP_{T}\) is at most \(\gamma_{T}\).
    To do this, we show that even if we weaken \(\PlayerAdversary\), the expected \(\ObjectiveBWMP\)-value of the outcome is at most \(\gamma_{T}\).
    Formally, we weaken \(\PlayerAdversary\) by replacing every vertex \(v \in \VerticesAdversary\) in \(\MDP_{T}\) that belongs to \(\PlayerAdversary\) with a probabilistic vertex with a uniform distribution over all out-neighbours of \(v\).
    This yields a Markov chain which we denote by \(\MC_{T}\).
    Since \(\MDP_{T}\) is an MEC, we have that \(\MC_{T}\) is a single bottom strongly-connected component (BSCC). 
    From~\cite{BGR19}, it follows that the expected \(\ObjectiveBWMP\)-value of the outcome in a BSCC is equal to the mean payoff of the cycle with the minimum mean payoff in the BSCC.
    Thus, in the MDP \(\MDP_{T}\) as well, \(\PlayerAdversary\) can ensure from every vertex in \(\MDP_{T}\) that the expected \(\ObjectiveBWMP\)-value of the outcome is at most \(\gamma_{T}\).

    Now, we show the other direction, that is, the expected \(\ObjectiveBWMP\)-value of every vertex in \(\MDP_{T}\) is at least \(\gamma_{T}\).
    Here, we strengthen \(\PlayerAdversary\) by replacing every probabilistic vertex in \(\MDP_{T}\) with a \(\PlayerAdversary\) vertex to obtain a non-stochastic one-player game \(\Game_{T}\). 
    We show that despite also having control over all probabilistic vertices, for all strategies of \(\PlayerAdversary\) in \(\Game_{T}\), the \(\ObjectiveBWMP\)-value of the outcome in \(\Game_{T}\) is at least \(\gamma_{T}\).
    Since \(\MDP_{T}\) is a MEC, we have that \(\Game_{T}\) is strongly connected and that \(\PlayerAdversary\) has a strategy to reach every vertex in \(\Game_{T}\) from every other vertex in \(\Game_{T}\).
    Since \(\ObjectiveBWMP\) is \kl{prefix-independent}, this implies that the \(\ObjectiveBWMP\)-value of every vertex in \(\Game_{T}\) is the same.
    For ease of analysis, we add \(- \gamma_{T}\) to every edge payoff in \(\Game_{T}\) to obtain a new one-player game \(\hat{\Game}_{T}\).
    We have that the \(\ObjectiveBWMP\)-value of vertices in the original game \(\Game_{T}\) is at least \(\gamma_{T}\) if and only if the \(\ObjectiveBWMP\)-value of vertices in the offset game \(\hat{\Game}_{T}\) is nonnegative.
    In the offset game \(\hat{\Game}_{T}\), every simple cycle has nonnegative mean payoff, and therefore, every simple cycle in \(\hat{\Game}_{T}\) also has nonnegative total payoff.
    We show that every \kl{play} in \(\hat{\Game}_{T}\) has \(\ObjectiveBWMP\)-value that is nonnegative.

    Let \(\PayoffFunction_{\min}\) denote the minimum edge payoff occurring in the offset game \(\hat{\Game}_{T}\).
    Since there is a simple cycle in \(\hat{\Game}_{T}\) with total payoff equal to zero, it cannot be the case that \(\PayoffFunction_{\min}\) is positive.
    If \(\PayoffFunction_{\min}\) is equal to zero, then we have that every edge in \(\hat{\Game}_{T}\) has nonnegative payoff, and therefore every \kl{play} in \(\hat{\Game}_{T}\) has \(\ObjectiveBWMP\)-value that is nonnegative.
    Now, suppose that \(\PayoffFunction_{\min}\) is strictly negative.
    Let \(\Play\) be a \kl{play} in \(\hat{\Game}_{T}\).
    Note that for every suffix \(\PlaySuffix{x}\) of \(\Play\), for every finite prefix \(\PlayInfix{x}{y}\) of \(\PlaySuffix{x}\), the total payoff of the segment \(\PlayInfix{x}{y}\) is bounded below by \(\abs{\Vertices} \cdot \PayoffFunction_{\min}\).
    Indeed, if the length of \(\PlayInfix{x}{y}\) is strictly less than \(\abs{\Vertices}\), then the total payoff of the segment is at least \(\abs{\PlayInfix{x}{y}} \cdot \PayoffFunction_{\min}\), which is greater than \(\abs{\Vertices} \cdot \PayoffFunction_{\min}\).
    Otherwise, if the length of \(\PlayInfix{x}{y}\) is at least \(\abs{\Vertices}\), then it contains a simple cycle with nonnegative total payoff. 
    Deleting this simple cycle from \(\PlayInfix{x}{y}\) gives a shorter \kl{play} whose total payoff is at most the total payoff of \(\PlayInfix{x}{y}\).
    Thus, for all \(\varepsilon < 0\), we have that all \(\varepsilon\)-windows in \(\Play\) close in at most \(\frac{\abs{\Vertices} \cdot \PayoffFunction_{\min}}{\varepsilon}\) steps, and thus, the \(\ObjectiveBWMP\)-value of the \kl{play} \(\Play\) is greater than \(\varepsilon\).
    Thus, the \(\ObjectiveBWMP\)-value of every vertex in the offset game \(\hat{\Game}_{T}\) is at least zero, and the \(\ObjectiveBWMP\)-value of every vertex in the one-player game \(\Game_{T}\) is at least \(\gamma_{T}\). 
    Hence, for all strategies of \(\PlayerAdversary\) in the MDP \(\MDP_{T}\), the expected \(\ObjectiveBWMP\)-value of the outcome in \(\MDP_{T}\) is also at least \(\gamma_{T}\).
    
    This shows that the expected \(\ObjectiveBWMP\)-value of every vertex in \(\MDP_{T}\) is equal to \(\gamma_{T}\) and concludes the proof of the observation.
\end{proof}
The following corollary of \Cref{lem:bwmp-values-without-boundary-vertices}, which is analogous to \Cref{cor:fwmpl-denominator-without-boundary-vertices},
states the bound \(\DenBoundNoBoundaryVertex{\BWMP}\) for the \(\ObjectiveBWMP\) objective.
\begin{corollary}%
\label{cor:bwmp-denominator-without-boundary-vertices}
    The expected \(\ObjectiveBWMP\)-value of vertices in \(\ValVectorBWMP\)-classes without \kl{boundary vertices} can be written as \(\frac{p}{q}\) where \(p\) and \(q\) are integers and  \(q \le \PayoffDenominator^{\abs{\Vertices}} \cdot \abs{\Vertices}\).
\end{corollary}
From \Cref{thm:denominator-bound}, we get that the denominator of \(\ClassValBWMP{i}\) of each class \(\BWMPClass{i}\) in \(\Game\) is at most \(2^{\abs{V}} \cdot \ProbabilityDenominator^{\abs{V}^3} \cdot (\DenBoundNoBoundaryVertex{\BWMP})^{\abs{\Vertices}}\), which is at most \( 2^{\abs{\Vertices}} \cdot \ProbabilityDenominator^{\abs{\Vertices}^3} \cdot (\PayoffDenominator^{\abs{\Vertices}} \cdot \abs{\Vertices})^{\abs{\Vertices}}\).
\begin{lemma}
    \label{lem:bwmp-value-denominator}
    The expected \(\ObjectiveBWMP\)-value of each vertex in \(\Game\) can be written as \(\frac{p}{q}\), where \(p, q\) are integers, and \(q \le 2^{\abs{\Vertices}} \cdot \ProbabilityDenominator^{\abs{\Vertices}^3} \cdot (\PayoffDenominator^{\abs{\Vertices}} \cdot \abs{\Vertices})^{\abs{\Vertices}}\), and \(-\ObjectiveBound{} \cdot q \le p \le \ObjectiveBound{} \cdot q\). 
\end{lemma}

\begin{remark}
    In~\cite{CDRR15}, it has been shown that for non-stochastic two-player games, there exists a large enough window length (\(\WindowLength_{\max} = (\abs{\Vertices} - 1) \cdot (\abs{\Vertices} \cdot \ObjectiveBound{} + 1)\), where \(\ObjectiveBound{}\) is maximum absolute edge payoff in the game) such that for all vertices \(v\) in the game, it is the case that \(v\) is winning for the \(\BWMP(0)\) objective if and only if it is winning for the \(\FWMP(\WindowLength_{\max}, 0)\) objective.
    We remark that in general, there does not exist a window length \(\WindowLength\) such that the expected \(\ObjectiveFWMPL\)-value of a vertex is equal to the expected \(\ObjectiveBWMP\)-value of the vertex.
    To see this, consider the game in \Cref{fig:bwmp-example-infinite-memory} (from~\cite{CDRR15}).
    We have that the expected \(\ObjectiveBWMP\)-value of both vertices in this game is zero.
    However, the expected \(\ObjectiveFWMPL\)-value of both vertices is equal to \(-1/\WindowLength\), which is strictly negative.
    \begin{figure}[t]
        \centering
        \begin{tikzpicture}
            \node[state] (v1) {\(\Vertex[1]\)};
            \node[square, draw, right of=v1] (v2) {\(\Vertex[2]\)};
            \draw 
                  (v1) edge[bend left] node[above, pos=0.3] {\(\EdgeValues{-1}{}\)} (v2)
                  (v2) edge[bend left] node[below, pos=0.3]{\(\EdgeValues{+1}{}\)} (v1)
                  (v2) edge[loop right] node[right, pos=0.3]{\(\EdgeValues{0}{}\)} (v2)
            ;
        \end{tikzpicture}
        \caption{expected \(\ObjectiveBWMP\) value of both vertices is zero, but the expected \(\ObjectiveFWMPL\)-value of both vertices is strictly negative for all \(\WindowLength \ge 1\). }
        \label{fig:bwmp-example-infinite-memory}
    \end{figure}
    \lipicsEnd
\end{remark}
We now state the main result of this section for the expected \(\ObjectiveBWMP\)-value problem.
\begin{theorem}%
\label{thm:bwmp-summary}
    The expected \(\ObjectiveBWMP\)-value problem is in \(\UP \intersection \coUP\). 
    Memoryless strategies suffice for \(\PlayerMain\). \(\PlayerAdversary\) requires infinite memory in general.
\end{theorem}
\begin{proof}
This proof follows a similar structure as the proof of \Cref{thm:fwmp-summary}.
As before, the \ConditionBellman\ condition can be checked in polynomial time. 
Checking the \ConditionLB\ and \ConditionUB\ conditions involves checking almost-sure satisfaction of the \kl{Boolean objective} \(\BWMP\) for appropriate thresholds, which reduces to checking the satisfaction of \(\BWMP\) in non-stochastic games~\cite{DGG25LMCS}, which in turn reduces to total supremum payoff~\cite{CDRR15}, which is in \(\UP \intersection \coUP\)~\cite{GS09}.
Both of these reductions are polynomial-time reduction, and hence, the expected \(\ObjectiveBWMP\)-value problem is in \(\UP \intersection \coUP\).

Memoryless strategies suffice for \(\PlayerMain\) for almost-sure satisfaction of \(\BWMP(\Threshold)\)~\cite{DGG25LMCS}.
\(\PlayerAdversary\) requires infinite memory in general for the \(\BWMP(\Threshold)\) objective even in non-stochastic games~\cite{CDRR15}, which are a special case of stochastic games.
Deterministic strategies suffice for both players.
Hence, we get the memory requirements of an optimal strategy for the expected \(\ObjectiveBWMP\)-value problem using \Cref{cor:memory-bound}.
\end{proof}

\section{Discussion}%
\label{sec:discussion}

We discuss some concluding remarks about the relation of our work to previous work~\cite{CHH09}, which deals with the satisfaction of Boolean \kl{prefix-independent} objectives.
We also discuss practical implementations for window mean-payoff objectives and applicability of our technique to other prefix-independent objectives.

\subparagraph*{Comparison with \cite{CHH09}.}
In~\cite{CHH09}, it suffices to check the almost-sure satisfaction of the same \kl{Boolean objective} \(\BooleanObjective\) in all value classes.
In contrast, for \kl{quantitative objectives}, the \kl{threshold Boolean objective} for which we check the almost-sure satisfaction depends on the guessed value of the value class
(``Can \(\PlayerMain\) satisfy \(\ThresholdObjective{> \ClassVal{i} - \GranularityBound{\Objective}}\) with probability~\(1\)?''). 
Another key difference is that for \kl{Boolean objectives}, the value classes without \kl{boundary vertices} 
are precisely the extremal value classes, that is classes with values~\(0\) and~\(1\). 
In the case of \kl{quantitative objectives}, there may be multiple intermediate value classes without \kl{boundary vertices}, making reasoning about the correctness of the reduction more difficult.

We note that if we apply our approach to Boolean \kl{prefix-independent} objectives (such as \Buchi, co\Buchi, parity) by viewing them as \kl{quantitative objectives} mapping each \kl{play} to \(0\) or \(1\), then we recover the algorithm given in~\cite{CHH09}.

\subparagraph*{Applicability to other prefix-independent objectives.}
Recall that in order to be able to apply our characterization to a prefix-independent objective \(\Objective\), we require bounds \(\ObjectiveBound{\Objective}\) on the image of \(\Objective\) and \(\DenBound{\Objective}\) on the denominators of the optimal expected \(\Objective\)-values of vertices in the game.
These bounds can easily be derived individually for the following quantitative prefix-independent objectives, and thus, our technique applies to these objectives. 

\begin{description}
    \item[Mean-payoff objective.]
    The mean-payoff value of any play is bounded between the minimum and maximum edge payoffs in the game, directly giving bounds on the image of the mean-payoff objective.
    Since deterministic memoryless optimal strategies exist for both players~\cite{EM79}, the expected mean-payoff value is a solution of a stationary distribution in the Markov chain obtained by fixing memoryless strategies of both players. 
    This gives denominator bounds for the expected mean-payoff values of vertices in the game.

    \item[Limsup and liminf objectives.]
      Since the \(\liminfObj\) objective is equivalent to \(\FWMPL\) objective with window length \(\WindowLength = 1\)~(\Cref{lem:NEXPcoNEXP}), and the \(\limsupObj\) objective is the dual of the \(\liminfObj\) objective, our analysis with window mean-payoff objectives generalizes \(\limsupObj\) and \(\liminfObj\) objectives.

    \item[Positive frequency payoff objective~\cite{GK23}.]
    Here, each vertex has a payoff, and this objective returns the maximum payoff among all vertices that are visited with positive frequency in an infinite play. 
    This objective is prefix-independent, as the frequency of a vertex is independent of finite prefixes. 
    The image of the objective is bounded between the minimum and maximum vertex payoffs. 
    We observe that the value of vertices in a value class without boundary vertices is equal to the payoff of a vertex in the class, (giving a denominator bound for these vertices), and Theorem 11 uses this to give a denominator bound for vertices in value classes with boundary vertices. 
\end{description}

\subparagraph*{Practical implementation.}
We discuss approaches to solve the expected \(\Objective\)-value problem for the window mean-payoff objectives in practice.

A trivial algorithm that works for both \(\ObjectiveFWMPL\) and \(\ObjectiveBWMP\) objectives is to iterate over all possible value vectors. 
For each value vector, we check if the conditions in \Cref{thm:six-conditions} are satisfied, which can be done in polynomial time. 
Since there are exponentially many possible value vectors, this algorithm has an exponential running time in the worst-case.

Another technique is value iteration~\cite{CH08}, which has been seen to be an anytime algorithm for the standard mean-payoff objective~\cite{KMW23}.
An anytime algorithm gives better precision the longer it is run, and can be interrupted any time.
Given a game \(\Game\) with \(|\Vertices|\) vertices, the expected \(\ObjectiveFWMPL\)-value problem on \(\Game\) reduces to the expected \(\liminfObj\)-value problem on a game \(\Game'\) with \(|\Vertices|^{\WindowLength}\) vertices, (that is, on an exponentially larger game graph).
The \(\liminfObj\) objective is a well-studied objective in the context of value iteration~\cite{CH08,CDH09}.
We describe the reduction in~\Cref{lem:NEXPcoNEXP}, which also gives the expected \(\ObjectiveFWMPL\)-values of vertices in \(\Game\).

Since the size of the graph \(\Game'\) is much bigger than that of \(\Game\), we would like to work with \(\Game'\) on-the-fly rather than explicitly constructing the entire graph.
In~\cite{KMW23}, the authors show bounded value iteration for objectives such as \kl{reachability} and mean-payoff. 
They also discuss that the algorithm can be extended to be asynchronous and use partial exploration. 
As future work, we would like to look at the practicality of on-demand asynchronous value iteration for the \(\liminfObj\) objective, or even the window mean-payoff objectives \(\ObjectiveFWMPL\) and \(\ObjectiveBWMP\) directly.
An interesting aspect of it would be to investigate heuristics and optimizations such as sound value iteration~\cite{QK18}, optimistic value iteration~\cite{HK20}, and topological value iteration~\cite{AEKSW22} to speed up the practical running time.

\bibliography{mybib}

\end{document}